\documentclass[journal]{IEEEtran}

\usepackage{subfigure}
\usepackage{pdfpages}
\usepackage{diagbox}
\usepackage{colortbl}
\definecolor{mygray}{gray}{0.85}
\usepackage{array}
\usepackage{caption}
\usepackage{flushend}  
\usepackage[linesnumbered,boxed,ruled,commentsnumbered]{algorithm2e}
 \usepackage{enumerate}
\usepackage{epstopdf}
\usepackage{array,multirow,diagbox}
\usepackage{amsmath,bm}
\usepackage{cite}
\usepackage{amsmath,amssymb,amsfonts}
\usepackage{algorithmic}
\usepackage{amsthm}
\usepackage{graphicx}
\usepackage{textcomp}
\usepackage{xcolor}

\begin{document}

\title{\Huge{Massive Access of Static and Mobile Users via Reconfigurable Intelligent Surfaces: \\Protocol Design and Performance Analysis}}

\author{
\IEEEauthorblockN{Xuelin Cao, Bo Yang, Chongwen Huang,
  George C. Alexandropoulos, \IEEEmembership{Senior Member,~IEEE},
  Chau Yuen,~\IEEEmembership{Fellow,~IEEE},
  Zhu Han, \IEEEmembership{Fellow,~IEEE}, 
  H. Vincent Poor,~\IEEEmembership{Life Fellow,~IEEE}, and 
  Lajos Hanzo,~\IEEEmembership{Fellow,~IEEE}
  }
\thanks{X. Cao is with the School of Cyber Engineering, Xidian
University, Xi’an 710071, China, and also with the Engineering Product Development Pillar, Singapore University of Technology and Design, Singapore 487372 (e-mail: xuelin\_cao@sutd.edu.sg).}
\thanks{B. Yang and C. Yuen are with the Engineering Product Development Pillar, Singapore University of Technology and Design, Singapore 487372 (e-mails: bo\_yang, yuenchau@sutd.edu.sg).}
\thanks{C. Huang is with College of Information Science and Electronic Engineering, Zhejiang University, Hangzhou 310027, China, and with International Joint Innovation Center, Zhejiang University, Haining 314400, China, and also with Zhejiang Provincial Key Laboratory of Info. Proc., Commun. \& Netw. (IPCAN), Hangzhou 310027, China. (e-mail: chongwenhuang@zju.edu.cn ).}
\thanks{G. C. Alexandropoulos is with the Department of Informatics and Telecommunications, National and Kapodistrian University of Athens, 15784 Athens, Greece (e-mail: alexandg@di.uoa.gr).}
\thanks{Z. Han is with the Department of Electrical and Computer Engineering in the University of Houston, Houston, TX 77004 USA, and also with the Department of Computer Science and Engineering, Kyung Hee University, Seoul, South Korea, 446-701. (e-mail: zhan2@uh.edu).}
\thanks{H. V. Poor is with the Department of Electrical and Computer Engineering, Princeton University, Princeton, NJ, 08544, USA. (e-mail: poor@princeton.edu).}
\thanks{L. Hanzo is with the School of Electronics and Computer Science, University of Southampton, Southampton, SO17 1BJ, UK. (e-mail: lh@ecs.soton.ac.uk).}
 }
\maketitle

\begin{abstract} 
The envisioned wireless networks of the future entail the provisioning of massive numbers of connections, heterogeneous data traffic, ultra-high spectral efficiency, and low latency services. This vision is spurring research activities focused on defining a next generation multiple access (NGMA) protocol that can accommodate massive numbers of users in different resource blocks, thereby, achieving higher spectral efficiency and increased connectivity compared to conventional multiple access schemes. In this article, we present a multiple access scheme for NGMA in wireless communication systems assisted by multiple reconfigurable intelligent surfaces (RISs). In this regard, considering the practical scenario of static users operating together with mobile ones, we first study the interplay of the design of NGMA schemes and RIS phase configuration in terms of efficiency and complexity. Based on this, we then propose a multiple access framework for RIS-assisted communication systems, and we also design a medium access control (MAC) protocol incorporating RISs. In addition, we give a detailed performance analysis of the designed RIS-assisted MAC protocol. Our extensive simulation results demonstrate that the proposed MAC design outperforms the benchmarks in terms of system throughput and access fairness, and also reveal a trade-off relationship between the system throughput and fairness.
\end{abstract}



\section{Introduction}

\IEEEPARstart{W}ith the envisioned demands for access by massive numbers of users, high spectral/energy efficiency (SE/EE), and low-cost services (e.g., virtual/augmented reality (VR/AR), holographic telepresence, etc.) for the forthcoming Sixth Generation (6G) networks, research in future wireless communications continues to focus on the design of next generation multiple access (NGMA). To improve the SE/EE and quality-of-service (QoS), the NGMA approaches have to overcome the limitations of multiple access schemes in current wireless network standards by leveraging the envisioned gains of emerging techniques, such as reconfigurable intelligent surfaces (RISs) and artificial intelligence (AI), as well as new technologies yet to be defined \cite{WSaad, KBLetaief, NKato}. These technologies enabling or being enabled by  NGMA also give impetus to the design of medium access control (MAC) protocols, involving joint communication, control, and computing functionalities, which are expected to be vital for highly efficient massive multiple access \cite{YLiu, XLiu, XCao}.

In the development process of multiple access technologies, conventional orthogonal multiple access (OMA) schemes allow each user's transmission in an orthogonal way, thereby simplifying the transceiver design and avoiding interference among users, such as time-division multiple access (TDMA), frequency-division multiple access (FDMA), and code-division multiple access (CDMA). Compared to the family of OMA schemes, the non-orthogonal multiple access (NOMA) schemes have already been investigated in 5G networks due to bringing an additional degree of freedom in the power domain, which help each user achieve its QoS target \cite{YChen, MVaezi, YLiu1, ZZhang}. However, both OMA and NOMA schemes have limitations. Specifically, for OMA, the multiple access efficiency is limited by the radio resources and the signaling overhead, while for NOMA, most research thus far has focused on static and broadband users, but without considering mobility and randomness\cite{YLiu}.

\subsection{Motivation}
Next generation wireless networks are rapidly evolving toward a distributed intelligent communication, sensing, and computing platform, realizing the software-based network functionality paradigm. Efficient NGMA schemes are required to adapt to this trend. In this unifying context, there are several challenges to be addressed by NGMA, with two of the most prominent being the following:
\begin{itemize}
\item \textbf{Challenge-1:} How to improve the throughput performance of NGMA approaches when operating in complex wireless environments, where there exist both static and mobile users?
\item \textbf{Challenge-2:} How to achieve improved connectivity via NGMA schemes, while guaranteeing access fairness?
\end{itemize}

Advances in metamaterials have recently fuelled research in RISs for beneficially reconfiguring wireless communication environments with the aid of large planar arrays of low-cost reconfigurable elements \cite{di2019smart, Liaskos, GCA0, SHu}, RISs are becoming a potential solution to tackle the above mentioned challenges. Facing a practical network consisting of both static and mobile users, the throughput performance and connectivity of traditional multiple access schemes are significantly affected by the randomness and mobility of users. For such cases, RISs can be incorporated into the NGMA design to enhance the wireless communication links of static and mobile users simultaneously, and thus improve the throughput and connectivity performance. However, RIS-assisted NGMA approaches will face the complex problem of the RIS design for massive numbers of static and mobile users, which refers to multi-user resource allocation and the RIS phase configuration optimization. 

Motivated by these potential advantages, in this paper, we investigate the intricate interplay between the MAC protocol and the RIS configuration, as highlighted in Fig. \ref{structure}. On the one hand, we demonstrate that the low-complexity of the RIS phase profile configuration and implementation significantly improves the MAC efficiency. On the other hand, our efficient MAC protocol supports the coexistence of static and mobile users in conjunction with our low-complexity RIS configuration, thereby supporting a massive number of users via multiple RISs. Based on these compelling features, an efficient NGMA scheme is indeed eminently suitable for next-generation wireless communication systems.

\begin{figure}[t]
  \captionsetup{font={small}}
\centerline{ \includegraphics[width=3in, height=1.5in]{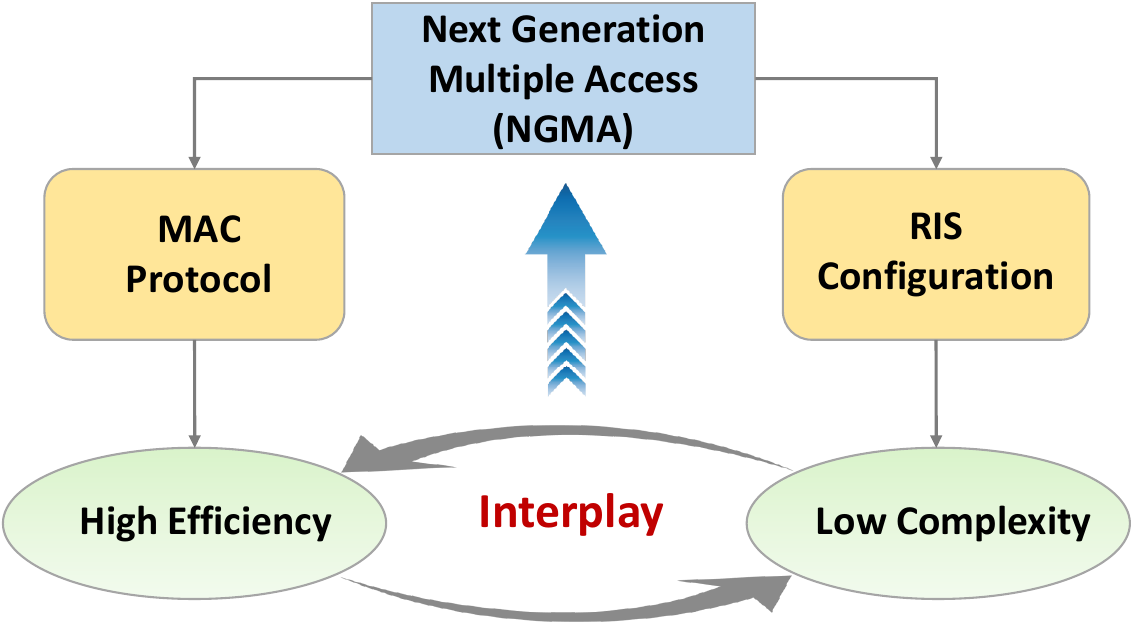}}
\caption{The interplay of MAC protocol and RIS configuration.}
\label{structure}
\end{figure}



\subsection{State-of-the-Art}
\label{related_work}
Recently, RISs have introduced some significant changes and new opportunities for wireless communications \cite{ Huang2020, Renzo, MA, GCA00}. This new paradigm results in the migration from traditional wireless connections to ``intelligent-and-reconfigurable connections".
\subsubsection{RIS Configuration in Wireless Communications}
Being a newly proposed paradigm going beyond massive multiple-input multiple-output (MIMO), RISs featured with low-cost, ultra-thin, light-weight, and low power consumption hardware structures provide a transformative means of wireless environments into a programmable smart entity. In the context of RIS-aided communications, the authors of \cite{wu2018intelligent, Huang2019, guo2020weighted, zhang2020reconfigurable, SLi, abeywickrama} focused their significant attention on the configuration of RISs. Explicitly, Wu \textit{et al.} studied the problem of joint active and passive beamforming \cite{wu2018intelligent}. To achieve high energy efficiency, Huang \textit{et al.} investigated an RIS-assisted downlink multi-user system by joint optimizing the transmit power and the passive beamforming \cite{Huang2019}. To increase the sum rate, Guo \textit{et al.} studied an RIS-aided multi-user multiple-input single-output downlink system by jointly designing the beamforming and RIS phase shifts \cite{guo2020weighted}. To assess the effect of RIS phase shifts on the data rate, Zhang \textit{et al.} jointly optimized the number of RIS phase shifts and RIS reflection beamforming for RIS-assisted communication systems \cite{zhang2020reconfigurable}. Li \textit{et al.} jointly designed the trajectory and the RIS reflect beamforming for UAV communications \cite{SLi}. Abeywickrama \textit{et al.} investigated a practical RIS phase shift model, and jointly designed the transmit beamforming and the RIS reflect beamforming \cite{abeywickrama}. In addition, the physical-layer security of RIS-assisted systems was analyzed in \cite{yu2019enabling}. Deep learning technologies in RIS-aided systems were explored in\cite{add2, GCA1, 9110869}. Deep learning technologies instead of conventional optimization methods were investigated for RIS-assisted aerial-terrestrial communications \cite{yang2020computation, XCao2}. 


\subsubsection{MAC Protocol for RIS-Assisted Communications}
With the development of the physical layer technological breakthrough on RISs, an enormous amount of research effort focus on the RIS-assisted multi-user communication system, especially its MAC protocol for system improvement. Until now, the distributed or the centralized MAC protocols for the system performance improvement have been proposed for RIS-assisted communications. To be specific, the TDMA-based scheme was present to enable multiple users' communications via RISs on the same frequency in different time slots. For example, Hu \textit{et al.} designed a frame-based MAC protocol for the RIS-assisted sensing system to achieve an accurate posture recognition \cite {JHu}. Bai \textit{et al.} proposed a TDD transmission protocol for RIS-aided mobile edge computing (MEC) systems \cite{TBai}. Cao \textit{et al.} proposed a frame-based MAC protocol to converge RIS and MEC into space information networks \cite {XCao3}, and Yang \textit{et al.} extended these investigations to discuss an RIS-assisted intelligent spectral learning system \cite{BYang}.

The FDMA-based scheme was adopted as well as by multiple users to communicate via RISs in the same time slot on non-overlapping domain frequency channels, and for example, Yang \textit{et al.} proposed a practical OFDM-based transmission protocol for RIS-enhanced communication system \cite{yang2020intelligent}. Jung \textit{et al.} also investigated the RIS-aided transmission protocol combined TDD with OFDMA schemes to achieve  user scheduling and power control \cite{MJung}. Moreover, the SDMA-based scheme was used to support communications among users via RISs either in a unique angular direction or by spatial multiplexing \cite{ZDing}. Furthermore, NOMA schemes were conceived for enhancing the multiple access performance of RIS-assisted multi-user communications \cite{ZDing1, LLV, XMu, WNi}. In contrast to these centralized MAC protocols, Cao \textit{et al.} designed a distributed MAC protocol for RIS-assisted multi-user system with considering mobility and randomness of users \cite{XCao, XCao1}. Additionally, efficient AI-based RIS-assisted MAC protocols have been investigated in \cite {XCao, XCao2}.



\subsection{Contributions and Organizations}


The major contributions are summarized as follows:
\begin{itemize}
\item 
\textbf{Framework design:} To improve the throughput performance of NGMA, we propose an RIS-assisted multiple access framework. In the proposed framework, the static and mobile users can communicate with the base station (BS) via RISs through different multiple access schemes at a low cost. 
 
\item 
\textbf{MAC protocol:} To achieve high connectivity and access fairness of NGMA, we design a MAC protocol that integrates the scheduled-based and contention-based schemes into a frame. By implementing different RIS configuration on different types of users, we achieve high efficient RIS-assisted multiple access, while considering randomness and mobility.  

\item 
\textbf{Analysis and optimization:} We first analyze the system throughput performance of the proposed MAC protocol. Then, we formulate a joint optimization problem to maximize the system throughput, while guaranteeing the fairness of users. To solve the formulated problem, we decompose the original problem into two sub-problems: the MAC design problem and the RIS phase configuration problem, and then an alternating optimization technique is adopted to solve them.

\item 
\textbf{Performance evaluation:} We evaluate the proposed MAC protocol in terms of system throughput and fairness. Simulation results reveal a trade-off relationship between the system throughput and fairness, and demonstrate that our MAC design outperforms benchmarks in terms of system throughput and access fairness.
\end{itemize}


The rest of this paper is organized as follows. We propose a multiple access framework for RIS-assisted multi-user communications in Section~\ref{MA framework}. We then design a MAC protocol for the proposed framework in Section~\ref{MA protocol}. Next, we analyze the system performance of the designed MAC in Section~\ref{MA performance}. Furthermore, we formulate a joint optimization problem that includes the MAC protocol and the RIS configuration to maximize the system performance, and we also solve the formulated mixed-integer nonlinear programming (MINLP) problem in Section~\ref{MA soving}. Simulation results are discussed in Section~\ref{result}. Finally, conclusions are drawn in Section~\ref{conclusion}.

\textit{Notations:} As per the traditional notation, a bold letter indicates a vector or matrix. $\max \{  \cdot \}$ and $\min \{  \cdot \}$ represent the maximum value and the minimum value, respectively. The amplitude of a complex number $x$ is denoted by $\left | x \right |$. The main notation we use is listed in Table \ref{Tab}.
 

\begin{table}[t]
		\small 
		\centering
			\renewcommand{\arraystretch}{1.1}
			\captionsetup{font={small}} 
			\caption{\scshape List of Main Notation} 
			\label{Tab}  
			\centering  
			\begin{tabular}{| c | l |}  
				\hline
				\textbf{Notation} & \textbf{Definition} \\
				\hline 
				$\mathcal{K}$ & The set of $K$ existing users \\ 
				\hline 
			    $\mathcal{M}$ & The set of $M$ RISs\\
			    \hline
			    $\mathcal{N}$ & The set of $N$ reflecting elements on one RIS\\
			    \hline
			    $\mathcal{C}$ & The set of $C$ sub-channels\\
			    \hline
			    $\mathcal{J}$ & The set of $J$ data slots on each sub-channel\\
			    \hline
			    $\mathcal{X}$ & The set of $X$ static users\\
			    \hline
			    $\mathcal{Y}$ & The set of $Y$ mobile users\\
			    \hline
			    $Z$ & The number of new mobile users\\
			    \hline
  		        $U_k$ & The $k$-th user\\
			    \hline
			    $R_m$ & The $m$-th RIS\\
			    \hline
			    ${\mathbf{g}}_{km}$ & The vector of reflected path between $U_k$ and $R_m$\\ 
				\hline 
			    ${\mathbf{h}}_{km}$ & The vector of reflected path between $R_m$ and BS\\
			    \hline
			    $r_k$ & The direct path between $U_k$ and BS \\
			    \hline
			    $D_{cj}$ & The $j$th data slot on sub-channel $c$\\
			    \hline
			    $u_k$ & The mobile profile of $U_k$\\
			    \hline
  		        $a_{km}$ & The state of $R_m$ for $X$ static and $Y$ mobile users\\
			    \hline
			    $t_{kj}$ & The state of $D_{cj}$ for $X$ static users\\
			   	\hline
			    ${\cal S}_s$ & The throughput of the scheduled transmissions\\
			    \hline
			    ${\cal S}_c$ & The throughput of the contended transmissions\\
			    \hline
			    ${\cal S}_0$ & The overall throughput\\
			    \hline
			    $\alpha$ & The ratio of the scheduled periods\\
			    \hline
			    $\beta$ & The ratio of the contended periods\\
			    \hline
			    $t_0$ & The duration of the pilot period\\
			     \hline
			    $t_1$ & The duration of the scheduled transmission period\\
			     \hline
			    $t_2$ & The duration of the computing transmission period\\
			    \hline
			    $t$ & The duration of a data slot\\
			    \hline
			    $\cal{T}$ & The set of $t_0,t_1$, and $t_2$\\
			    \hline
			    $s_k$ & The transmit signal of $U_k$\\
			    \hline
			    $w_k$ & The additive white Gaussian noise\\
			    \hline
			    $\mathbf{\Theta}_{km}$ & The matrix of RIS reflection coefficient of $U_k$\\
			    \hline
			    $\mathbf{\Psi}$ & The matrix of RIS phase shift \\
			    \hline
			    $\bm{\theta}_{km}$ & The vector of phase shift on $R_m$ of $U_k$\\
			    \hline
			    $\text{SNR}_{km}$ & The SNR at the BS from $U_k$ via $R_m$\\
			    \hline
			    $\rho_k^2$ & The transmit power of $U_k$\\
			    \hline
			    $B$ & The total bandwidth\\    
			\hline
			\end{tabular}  
				\end{table}

\section{Considered Multiple Access framework}
\label{MA framework}
In this section, we first introduce the system scenario of RIS-assisted multi-user wireless communications in Section \ref{SS}, and then we present a multiple access framework in Section \ref{FD}.

\subsection{System Scenario}
\label{SS}

We explicitly consider the different mobility profiles of users in practical scenarios (e.g., in smart industries where fixed sensors and mobile robots co-exist.), and improve the efficiency of the MAC protocol and reduce the complexity of the RIS configuration in our multi-user communication scenario by exploiting users' mobility profiles, as illustrated in Fig.~\ref{sys_model}.  In contrast to \cite{ZDing2, XMu1}, the users having similar mobility profiles are grouped together to enhance the interplay of the RIS configuration and MAC protocol. In our scenario, we consider $K$ existing users, $Z$ new mobile users, $M$ RISs each having $N$ reflecting elements, and a BS, where $M$ RISs are employed to assist the communications of $K+Z$ users with the BS over $C$ sub-channels. Here, we define the existing users and new mobile users as follows.
\begin{itemize}
\item \textit{Existing users:} These are the users who are already supported by the network. There are two types of existing users: static users without mobility and mobile users who may move out of the BS's coverage area.
\item \textit{New mobile users:} These are users who are just joining the network in the current frame due to mobility.
\end{itemize}

\par The set of existing users, RISs, reflecting elements and sub-channels are denoted by ${\cal K}=\{1, \ldots, k, \ldots, K\}$, ${\cal M}=\{1,\ldots,m,\ldots,M\}$, ${\cal N}=\{1,\ldots,n,\ldots,N\}$, and ${\cal C}=\{1,\ldots,c,\ldots,C\}$, respectively. We denote the $k$th user as $U_k, \ k \in \bar{{\cal K}}$, where $\bar{{\cal K}}=\{1,\ldots, k, \ldots, K+Z\}$, while we represent the $m$th RIS as $R_m, \ m \in {\cal {M}}$. Each user is equipped with a single antenna and the BS is equipped with multiple antennas. Each RIS is equipped with a controller connected to the BS. The vector of reflected path between $U_k$ and $R_m$ is denoted by ${\mathbf{g}_{km}}\in \mathbb{C}^{N\times1}$, the vector of reflected path between $R_m$ and BS for $U_k$ is denoted by ${\mathbf{h}}_{km}\in \mathbb{C}^{1\times N}$, and the direct path between $U_k$ and BS is denoted by $r_k$, where $k\in \bar{{\cal K}}, m\in \mathcal{M}$. In the system considered, a quasi-static fading channel model and perfect channel state information (CSI) are assumed\footnote{Channel estimation in RIS-assisted wireless systems is an ongoing field of research with approaches ranging from cascade channel estimation via passive RISs to compressed sensing channel estimation via RISs with minimal and basic sensing capabilities \cite{GCA, SLin}.}. Each RIS is assumed to be equipped with passive elements and operate in the non-overlapping frequency domain\footnote{Compared with the desired reflected signal, the interference power caused by reflections via the remaining RISs in non-overlapping frequency bands is relatively low \cite{SV}, and can be ignored.}. Additionally, we assume that one RIS can be used by one user at a time, and each user has the same payload. We also assume that the BS knows the network state of all existing users\footnote{The new mobile users will be regarded as an existing user in the next frame once it has joined in the current frame, and the BS will update the value of $K$ based on the dynamic network.}.

\begin{figure}[t]
  \captionsetup{font={small}}
\centerline{ \includegraphics[width=3.4in, height=2.7in]{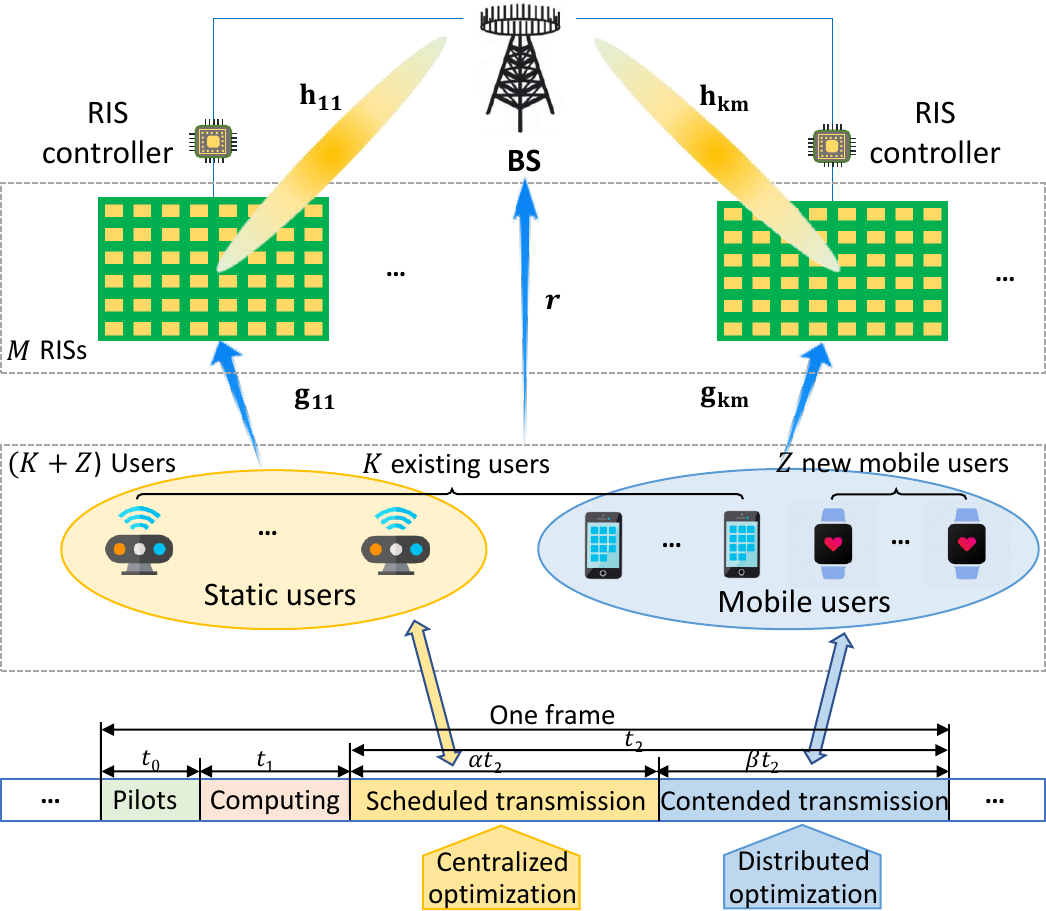}}
\caption{The multiple access framework for RIS-assisted communications by exploiting users' mobility profiles, where $K+Z$ users communicate with a BS via $M$ RISs. During a frame, based on pilot transmission and computation, static users communicate with the BS by scheduling $M$ RISs, while mobile users communicate with the BS using these RISs by contention. The ratio of the transmission period of two types is $\beta/\alpha$, which is optimized by the BS during the computing period.}
\label{sys_model}
\end{figure}

\subsection{Framework}
\label{FD}
A multiple access framework is conceived for RIS-assisted communications to exploit the symbiotic interplay between the MAC protocol and RIS configuration. The proposed multiple access framework is shown as in Fig. 2, which combines two aspects: 1) the MAC protocol, and 2) the configuration mode of multiple RISs. These two interact with each other according to the dynamic wireless environments. The MAC protocol is designed and optimized on a frame-by-frame basis, and each frame is divided into three periods: a pilot period, $t_0$; a computing period, $t_1$; and a transmission period, $t_2$. The latter consists of the scheduled and the contended transmission periods, the proportion of each is $\alpha$ and $\beta$, respectively. In a frame, based on pilot transmissions and optimizing computation, the proposed MAC protocol switches between the scheduled and the contended modes (i.e., static users are scheduled to communicate with the BS via $M$ RISs, while mobile users are allowed to contend for communications with the BS via the same RISs.). If only static or only mobile users have to be served in a frame, the proposed MAC protocol will become a pure scheduled mode or a pure contended mode. The ratio, $\beta/\alpha$ is optimized by the BS during the computing period. In addition, centralized optimization is used for scheduled transmissions, while distributed optimization is used for contended transmissions. In contrast to the conventional RIS-aided MAC design methods \cite{XCao1}, the proposed multiple access framework incorporates the RIS configuration into the MAC protocol, thereby improving the MAC efficiency as well as reducing the complexity of RIS configuration.

Specifically, the proposed multiple access framework has to achieve the following two challenging goals:
\begin{itemize}
\item For the MAC protocol, the low-complexity RIS configuration helps with improving the system throughput performance of the MAC protocol, while guaranteeing the fairness of users. 
\item For the RIS configuration, an efficient MAC protocol can decrease the RIS configuration complexity and improve the RIS utilization, and thus the system throughput performance of the MAC protocol can be further enhanced. 
\end{itemize}
\begin{figure*}[t] 
\centering
\captionsetup{font={small }}
\includegraphics[width=6.3in,height=1.45in]{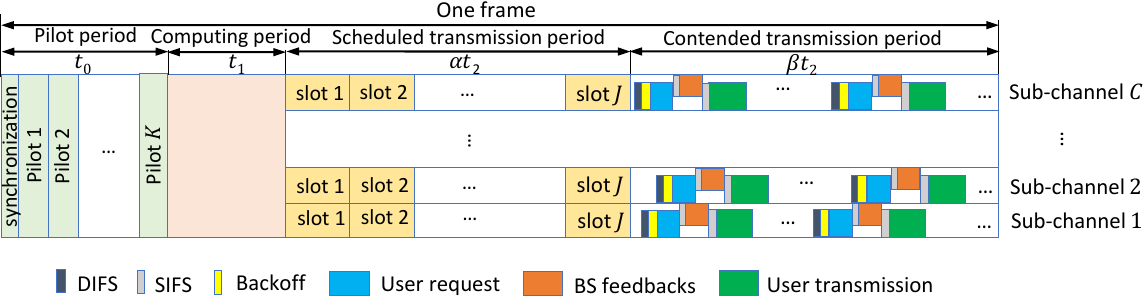}
\hspace{-1.5mm}
\caption{The MAC protocol is a frame-based structure, which integrates the scheduling access and the contention access into a frame.} 
\label{PD}
\end{figure*}
\vspace{-2mm}

\section{Proposed RIS-Assisted MAC Protocol}
\label{MA protocol} 
 \theoremstyle{Remark}
\newtheorem{remark}{\textbf{Remark}} 
\vspace{-1mm}
Based on the proposed multiple access framework and its design objectives, in this section, we present our scheme in terms of two aspects: i) MAC protocol in Section \ref{MACD}, and ii) RIS configuration in Section \ref{RISD}. Then, an intuitive example is illustrated in Section \ref{AIE}.

\subsection{MAC Protocol}
\label{MACD}
Based on the proposed multiple access framework in Section \ref{FD}, a MAC protocol that integrates the centralized and distributed implementations into a frame is presented, as illustrated in Fig. \ref{PD}. The pilot period is further divided into $K$ pilot slots according to the number of existing users. Based on the received pilot transmissions, the BS estimates the CSI of static users and calculates their RIS configuration and MAC protocol parameters during the computing period. Assume that the scheduled transmission period contains $J$ data slots, and the $j$th data slot on sub-channel $c$ is denoted as $D_{cj} , \ c \in {\cal {C}}, \ j \in {\cal {J}}$, where ${\cal J}=\{1,\ldots, j,\ldots,J\}$ is the set of data slots on each sub-channel. Then the BS tightly coordinates the multiple access of static users, following that, the static users transmit their data to the BS via RISs based on the scheduled results. During the contended transmission period, mobile users compute the RIS configuration and contend for their access based on their CSI.

Specifically, after the pilot transmissions and computing, the static users are scheduled, while the mobile users (i.e., the existing mobile users and the new mobile users in the current frame) are allowed to contend. In this context, the implementation of the RIS-assisted MAC protocol goes through the following four steps.

\par \textbf{Step 1: Pilot transmissions}. Based on the synchronization of the BS, during the pilot period, $K$ existing users transmit their pilots to the BS for their RIS-assisted transmissions. 

\par \textbf{Step 2: Computing and feedback}. During the computing period, the BS has to carry out the computations as follows.
\begin{itemize}
\item \textit{User classification}. After receiving the pilot transmission of users, the BS classifies these users according to the known network information. The type of a user is defined as $u_k, \forall k \in {\cal K}$, which is denoted by 
\begin{equation} \label{U} 
u_k = \begin{cases}
1,    & \text{if}\ U_k \text{ is static user,} \ \forall k \in {\cal K}, \\ 
0, & \text{if}\ U_k \text{ is mobile user,} \ \forall k \in {\cal K}.
\end{cases}
\end{equation}  
Thus, the number of static users is denoted as $X=\sum_{k=1}^{K}u_k$. Let ${\cal X} =\{1,\ldots, k, \ldots X\}$ be the set of $X$ static users. 


\item \textit{Channel estimation}. According to the user classification, the BS estimates the involved links of static users based on its pilot transmission on each sub-channel, i.e., $\mathbf{g}_{km}$, $\mathbf{h}_{km}$, and $r_k$, where $k\in \mathcal{X}, m\in \mathcal{M}$. 

\item \textit{Resource allocation}. Based on the channel information of static users, the BS computes the MAC protocol parameters and allocates the slots, power, and RIS resources for static users. To be specific, the BS first computes the duration of the scheduled transmission period and that of the contended transmission period, which is depicted in Section \ref{MA performance}. Then, the BS allocates $M$ RISs and $J$ data slots to static users over $C$ non-overlapping sub-channels. Since each user is only allowed to use one RIS in a frame, we define the state of $M$ RISs for $X$ static users as
\begin{equation} \label{A} 
a_{km} = \begin{cases}
1,    & \text{if}\ U_k \leftarrow R_m, \ \forall k \in {\cal X}, m \in {\cal {M}}, \\ 
0, & \text{Otherwise}, \ \forall k \in {\cal X}, m \in {\cal {M}},
\end{cases}
\end{equation}
where $U_k \leftarrow R_m$ means that $R_m$ is allocated to $U_k$. We then define the state of $J$ data slots for $X$ static users as
\begin{equation}  \label{T} 
t_{kj} = \begin{cases}
1,    & \text{if} \ U_k \leftarrow D_{cj}, \ \forall k \in {\cal X}, j \in {\cal {J}},  c \in {\cal {C}}, \\ 
0, & \text{Otherwise}, \ \forall k \in {\cal X}, j \in {\cal {J}}, c \in {\cal {C}},
\end{cases}
\end{equation}
where $U_k \leftarrow D_{cj}$ means that $D_{cj}$ is allocated to $U_k$, we have $c=m$ since each RIS is bonded to a sub-channel.
\item \textit{RIS phase configuration}. The reflection parameters of $M$ scheduled RISs are computed at the BS to support the transmission of $X$ static users. Note that the computation of RIS configuration should be considered with the MAC protocol parameters, such as $\alpha$, $\beta$, and $t_2$.
\end{itemize}

Based on the above computations, the BS feedbacks the scheduled information to static users during the computing period. 

\par \textbf{Step 3: Scheduled transmission for static users}. During the scheduled transmission period, the BS instructs each RIS controller to configure its reflection parameters and initiates the RIS-assisted transmissions of the scheduled static users. 

\par \textbf{Step 4: Contended transmission for mobile users}. As the designated contended transmission period begins, the unscheduled mobile users (i.e., $K-X$ existing mobile users and $Z$ new mobile users) start their multiple access and compute the RIS configuration by themselves based on the estimated CSI, which is calculated according to the sensing of each mobile user on each sub-channel. In contrast to the scheduled transmission of static users, mobile users have to negotiate with the BS for their channel access and RIS configuration, which is based on the distributed coordination function (DCF) scheme. Here, the results that $Y$ mobile users contend for $M$ RISs can be denoted by 
\begin{equation} \label{A} 
a_{km} = \begin{cases}
1,    & \text{if}\ U_k \rightarrow R_m, \ \forall k \in {\cal Y}, m \in {\cal {M}}, \\ 
0, & \text{Otherwise}, \ \forall k \in {\cal Y}, m \in {\cal {M}},
\end{cases}
\end{equation}
where $U_k \rightarrow R_m$ means that $U_k$ contends for $R_m$ successfully.  ${\cal Y} =\{1,\ldots, k, \ldots Y\}$ is the set of $Y$ mobile users, and $Y=K-X+Z$.

Specifically, the involved four actions at mobile users and the BS are described as follows.

\begin{itemize}
\item \textit{Sensing, computing, and Backoff}. A mobile user senses the state of $C$ sub-channels. Once some sub-channels are sensed to be idle, the mobile user computes the configuration of unused $\hat M (\hat M \le M)$ RISs and occupies an optimal RIS and the according sub-channel, then starts the backoff after a DCF inter-frame space (DIFS).
 
\item \textit{Request}. Once the backoff is finished at the contended sub-channel, the mobile user sends a request-to-send (RTS) packet included its RIS configuration information to the BS on its occupied sub-channel.

\item \textit{Feedback}. If the requested RIS is available for the mobile user on its occupied sub-channel, the BS allows the RIS controller to configure the RIS reflection parameters, and replies a clear-to-send (CTS) packet to the mobile user after a short inter-frame spacing (SIFS).

\item \textit{Transmission}. Following the elapse of a SIFS, once the CTS feedback is received at the mobile user, the mobile user then transmits its data to the BS via the configured RIS on its occupied sub-channel. 
\end{itemize}

Given the dynamic switching between the scheduled mode and the contended mode, our MAC protocol is capable of maintaining the target rate via RISs at a low cost. Additionally, the fairness of the static and mobile users having poor channel conditions can be maintained by scheduling and effective contention, respectively.

By implementing the designed MAC protocol, the overall system throughput is calculated by 
\begin{equation} \label{So} 
{\cal S}_o = \frac{t_2}{t_0+t_1+t_2} (\alpha{\cal S}_s+\beta{\cal S}_c ),
\end{equation}
where ${\cal S}_s$ and ${\cal S}_c$ indicate the throughput of the scheduled transmissions and that of the contended transmissions, respectively. $\alpha$ and $\beta$ are the ratio of the scheduled transmission period and contended transmission period, respectively, and $\alpha+\beta=1, \alpha\in[0,1], \beta\in[0, 1]$. In addition, $t_0,t_1$, and $t_2$ denote as the pilot, the computing, and the transmission periods, respectively, denoted by the set ${\cal T}=\{t_0,t_1,t_2\}$. Due to the introduction of RISs, the overall system throughput in \eqref{So} is not only affected by the MAC protocol but also significantly affected by the RIS configuration.

Introducing RISs enhances the quality of wireless links, thereby attaining the following benefits of the MAC protocol.
\begin{itemize}
\item  \textit{MAC efficiency improvement}. On the one hand, a low-complexity RIS configuration will reduce the computing cost; On the other hand, the access latency of static and mobile users can be reduced by the low-complexity scheduling and the low-contention collision, respectively. Thus, the efficiency of the MAC protocol is improved.
\item  \textit{MAC fairness improvement}. Because of the separation of static and mobile users and the low-complexity operation on them, the fairness of users can be enhanced.
\end{itemize}

 \subsection{RIS Configuration}
 \label{RISD}
As illustrated in Fig. \ref{PD}, the RIS configuration is integrated into the computing period for static users and the contended transmission period for mobile users. By separately configuring RIS phase shifts at the BS and mobile users for the scheduled and contended transmissions, the complexity of the RIS configuration is decreased.

In the system, the user-RIS-BS channel is thus modeled as a composition of two components, namely, the direct path (i.e., the user-BS link) and the reflect path (i.e., the user-RIS-BS path including the user-RIS link and the RIS-BS link). Hence, the signal received at the BS from $U_k$ through both the user-RIS and user-RIS-BS channels is denoted by
\begin{equation}\label{received signal1}
\centering
\!y_k\! =\! \underbrace{r_{k}s_k}_{\text{direct path}}\!+\!\underbrace{\mathbf{h}_{km} \mathbf{\Theta}_{km} \mathbf{g}_{km} s_k}_{\text{reflect path}}+w_k, k \in  \bar{{\cal K}}, \ m \in {\cal {M}}, 
\end{equation}
where $s_k$ represents the transmit signal of $U_k$, and it is an independent random variable with zero mean and unit variance (normalized power). $w_k$ denotes the additive white Gaussian noise (AWGN) at the BS, $w_k \sim {\cal CN} (0, \sigma^2$). $r_{k}$ is the direct link when $U_k$ transmits data to the BS.


In \eqref{received signal1}, $\mathbf{\Theta}_{km}$ is the matrix of the RIS reflection coefficient of $U_k$, which can be expressed as
\begin{equation}\label{channel1}
\centering
\!\mathbf{\Theta}_{km}\! =\!\text{diag}\left(\phi_{km}^{1},\!\ldots,\phi_{km}^{n},\!\ldots,\phi_{km}^{N}\right), \!k \in  \bar{{\cal K}}, \ m \in {\cal {M}}, 
\end{equation}
where $\!\phi_{km}^n\!=\!\gamma_{km}^ne^{j\theta_{km}^n}$ is the reflection coefficient of RIS element $n$ on $R_m$ for $U_k$,  $\{\theta_{km}^{n},\gamma_{km}^{n}\}$ are the phase shift and amplitude reflect coefficient of RIS element $n$ on $R_m$ for $U_k$. In practice, we assume that a continuous phase shift with a constant amplitude reflection coefficient is applied to each RIS element, i.e., $\vert\gamma_{km}^{n}\vert=1, \theta_{km}^n \in [0, 2\pi), k \in  \bar{{\cal K}}, \ m \in {\cal {M}}, \ n \in {\cal {N}}$. Let $\mathbf{\Psi}=[\bm{\theta}_{11},\ldots, \bm{\theta}_{km}, \ldots, \bm{\theta}_{(K+Z)M}]\in \mathbb{C}^{N\times (K+Z)M}$ denote the matrix of the RIS phase shift, where $\bm{\theta}_{km}\!=\![\theta_{km}^1, \ldots, \theta_{km}^n, \ldots, \theta_{km}^N]^T\in \mathbb{C}^{N\times1}$ is the vector of the phase shift of $R_m$ that is aligned to $U_k$. 

Accordingly, the SNR at the BS from $U_k$ via $R_m$ is expressed as
\begin{equation}\label{SNRE1}
\centering 
\!\text{SNR}_{km} = \!\left|\left(r_{k}\!+\!\mathbf{h}_{km} \!\mathbf{\Theta}_{km} \!\mathbf{g}_{km}\right)\!\rho_k\right|^2\!/\!\sigma^2, k \!\in \!\bar{{\cal K}},  m \!\in {\!\cal {M}},
\end{equation}
where $\rho_k^2$ is the transmit power of $U_k$.

As a benefit of the proposed MAC protocol, the RIS configuration has the following three advantages.
\begin{itemize}
\item \textit{Complexity reduction}. By implementing centralized and distributed operations for static and mobile users instead of harnessing centralized operation for all users, the computational complexity of RISs can be substantially reduced in the computing period and the contended transmission period, respectively. 
\item \textit{RIS utilization improvement}. Upon considering dynamic wireless environments, each period of a frame can be adjusted to improve the utilization of RISs. 
\item \textit{Service fairness improvement}. The RISs can serve new mobile users in the contended period of the current frame, thereby providing fairness for the users.
\end{itemize}
\begin{figure}[t]
\captionsetup{font={small }}
\centerline{ \includegraphics[width=3in, height=3.5in]{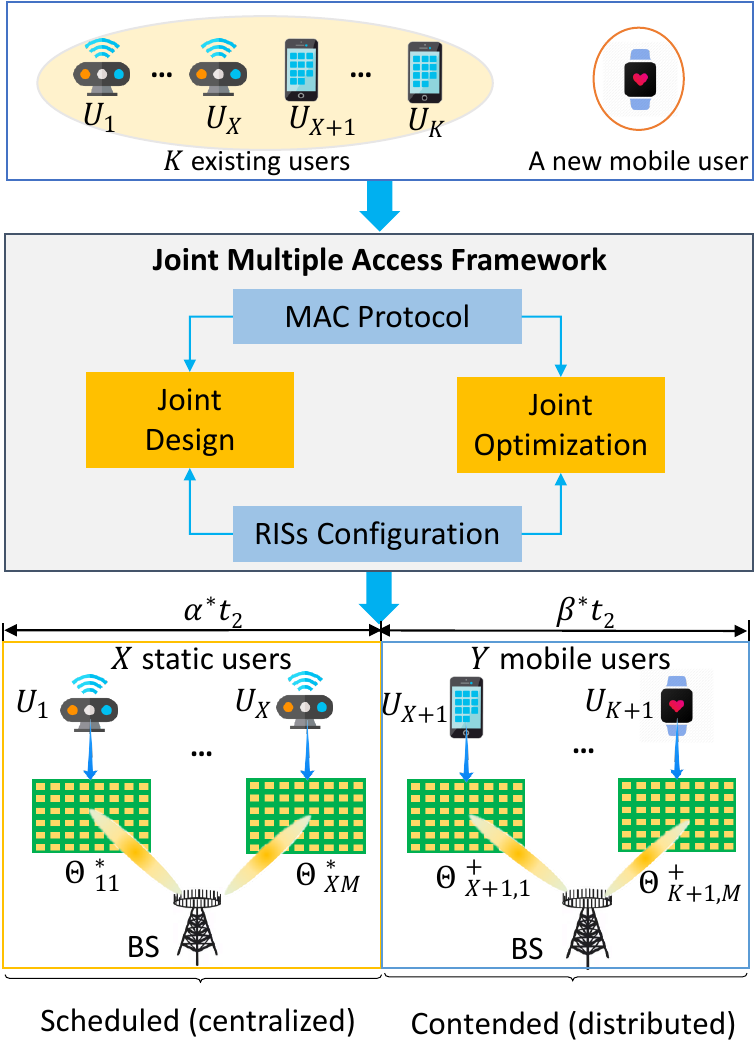}}
\caption{An intuitive example of the proposed framework.}
\label{case}
\end{figure}

\subsection{An Intuitive Example}
\label{AIE}
An intuitive implementation example with $K+1$ users is illustrated in Fig.~\ref{case}, where $K+1$ users include $K$ existing static and mobile users, and a new mobile user that joins the system in the current frame.

\begin{table}[t] 
\newcommand{\tabincell}[2]{\begin{tabular}{@{}#1@{}}#2\end{tabular}}
		\small
		\centering
			\renewcommand{\arraystretch}{1.1}
			\caption{\scshape Summary of design and optimization} 
			\label{table0}
			\footnotesize
			\centering  
			\begin{tabular}{ c|c|c|c}  
				\hline
               \multicolumn{2}{c|}{\textbf{Multiple access framework}} &\textbf{Static users}& \textbf{Mobile users}\\ 
\hline
	           \multirow{2}{*}{Design} & MAC protocol & Scheduled & Contended\\ 
			   \cline{2-4}
			   \multirow{2}{*}{} & RIS configuration & Centralized & Distributed\\ 
			   \cline{2-4}
			   \hline
			   \multirow{2}{*}{Optimization} & MAC protocol& $\alpha^*, {\cal T}^*$ & $\beta^*, {\cal T}^*$\\  
			   \cline{2-4}
			   \multirow{2}{*}{} & RIS configuration & $\!\mathbf{\Theta}_{km}^*, k\in \mathcal{X}$ & $\!\mathbf{\Theta}_{km}^+, k\in \mathcal{Y}$\\ 
			   \cline{2-4}
			    \hline		   
			\end{tabular}  
	\end{table}
	
In the illustrated example, by implementing the proposed multiple access framework, $K$ existing users are classified into two types, i.e., $X$ static users and $K-X$ mobile users. During the computing period, the scheduled transmission period and the contended transmission period are optimized based on the classification of users, then the BS schedules $X$ static users and $M$ RISs for their RIS-assisted communications in the scheduled transmission period. Next, the existing $K-X$ mobile users and a new mobile user contend for their RIS-assisted communications during the contended transmission period. Based on the optimized RIS configuration in the different transmission periods, each RIS controller is controlled by the BS to support the static users and the mobile users, respectively.  A summary of design and optimization in the multiple access framework, as illustrated in Table~\ref{table0}. 


\section{Performance Analysis}  
\label{MA performance}
In this section, we thoroughly analyze the system throughput performance of the proposed multiple access framework, where the RIS-assisted scheduled transmissions are analyzed in Section \ref{ST}, and the RIS-assisted contended transmissions are analyzed in Section \ref{CT}.
 
\subsection{RIS-Assisted Scheduled Transmissions}
\label{ST}
For the proposed multiple access framework, because of employing TDMA and FDMA schemes for static users, the system throughout of RIS-assisted scheduled transmissions, ${\cal S}_s$, is the sum throughput of each scheduled static user. Thus, ${\cal S}_s$ can be expressed as
\begin{equation}\label{Ss}
\centering 
{\cal S}_s =\frac{1}{\alpha t_2}\sum_{k=1}^{X}\sum_{m=1}^{M}\sum_{j=1}^{J}a_{km}t_{kj}t\frac{B}{C}\log_{2}\left(1+\text{SNR}_{km}\right),
\end{equation} 
where $B$ is the total bandwidth, $\sum_{m=1}^{M}a_{km}=1,\forall k\in \cal X $, and $\sum_{j=1}^{J}t_{kj}=1,\forall k\in \cal X$, $t$ is the duration of a data slot.


According to \eqref{SNRE1}, \eqref{Ss} can be rewritten as
\begin{equation}\label{Ss1}
\centering 
\!{\cal S}\!_s \!\!=\!\!\frac{Bt}{\!C \alpha  t_2}\!\sum_{k=1}^{X}\!\sum_{m=1}^{M}\!\sum_{j=1}^{J}\!a_{km}\!t_{kj}\!\log\!_{2}\!\left(\!1\!+\!\!\frac{\!\left|\!\left(r_{k}\!+\!\mathbf{h}_{km}\!\mathbf{\Theta}_{km}\mathbf{g}_{km}\!\right)\!\rho_k\!\right|^2}{\!\sigma^2}\!\right)\!,
\end{equation} 
where $\alpha t_2 = Jt$, $\sum_{m=1}^{M}a_{km}=1,\forall k\in \cal X$, and $\sum_{j=1}^{J}t_{kj}=1,\forall k\in \cal X$.

To guarantee the RIS-assisted transmission of each static user, the scheduled transmission period has to meet the following condition,
\begin{equation}\label{J}
\centering 
JC \ge X.
\end{equation}  

From \eqref{Ss1}, it can be observed that the system throughput of RIS-assisted scheduled transmissions is mainly affected by the configuration of RISs. Because the benefit of RISs, the data rate of each scheduled static user is increased significantly. The higher data rate is thereby improving the system throughput of RIS-assisted scheduled transmissions when the value of $\alpha t_2$ is given. In addition, the system throughput of RIS-assisted scheduled transmissions is also improved by optimizing the MAC protocol parameters ($\alpha$ and $t_2$), i.e., using the minimum time to achieve the RIS-assisted transmission of all static users.

\subsection{RIS-Assisted Contended Transmissions}
\label{CT}
Based on the performance analysis of RIS-assisted scheduled transmissions, $X$ static users are scheduled by the BS to use $M$ RISs in the duration of $\alpha t_2$, while the other $K-X$ mobile users, as well as $Z$ new mobile users, are allowed to contend for $M$ RISs in the duration of $\beta t_2$. In other words, total $Y=K-X+Z$ mobile users contend for $M$ RISs in the duration of $\beta t_2$ for their RIS-assisted transmissions. Thus, the performance of RIS-assisted contended transmissions is analyzed as follows.


During the contended transmission period, we assume that each idle sub-channel has an equal probability, $\frac{1}{C}$, to be the best idle channel sensed by a mobile user, and the probability that $V_i$ mobile users select a given channel, $P_{V_i}$, can be given by 
\begin{equation}\label{Pn}
\centering 
\!P_{V_i}\!=\!\left(\!\begin{array}{c} \!N_i \!\\ \!V_i \!\end{array}\!\right)\!\left( \!\frac{1}{C} \!\right)^{V_i}\left(1\!-\! \frac{1}{C} \right)^{N_i\!-\!V_i},
\end{equation} 
where $N_i$ is the number of mobile users that contend for their RIS-assisted transmissions at the $i$th time, and $V_i$ is the number of mobile users that select a given channel at the $i$th time.

Refer to \cite{Bianchi}, at the $i$th time, as $V_i$ mobile users contend for their access at the particular sub-channel, the successful transmission probability ($P_{i,s}$), the idle probability ($P_{i,e}$), and the failed transmission probability ($P_{i,c}$) can be expressed as 
\begin{equation}\label{q1}
\centering
\left\{\begin{array}{l}
P_{i,s}=V_i\tau_i\left(1-\tau_i\right)^{V_i-1}, \\
P_{i,e}=(1-\tau_i)^{V_i-1},\\
P_{i,c}=1-P_{i,e}-P_{i,s},
\end{array} \right.
\end{equation}
where $\tau_i$ is the stationary probability that a mobile user transmits a data packet in a random slot at the $i$th time, which is expressed as
\begin{equation}\label{q2}
\centering
\tau_i=\frac{2(1-2p_i)}{(1 - 2p_i)(W+1) + p_iW(1 - {(2p_i)}^l)}.
\end{equation}
In \eqref{q2}, $W\in[W_{min},W_{max}]$ is the contention window and $l$ is the backoff stage. $W_{min}$ and $W_{max}$ denote the minimum and maximum contention window, respectively. $p_i$ is the collision probability at the $i$th time, which is calculated by  
\begin{equation}\label{q3}
\centering
p_i=1-(1-\tau_i)^{V_i-1}.
\end{equation}

\begin{figure*}[t]
\setcounter{equation}{31}
\begin{equation} \label{So1} 
\!{\cal S}\!_o \!= \!\frac{\!B(t\!+\!t_d)}{\!C\!(\!t_0\!+\!t_1\!+\!t_2\!)}\!\sum_{m=1}^{M} \!\left(\sum_{k=1}^{X}\!\sum_{j=1}^{J}\!a_{km}\!t_{kj}\!\log_{2}\!\left(\!1\!+\!\frac{\!\left|\left(r_{k}\!+\!\mathbf{h}_{km} \mathbf{\Theta}_{km} \mathbf{g}_{km}\!\right)\rho_k\right|^2}{\sigma^2}\!\right) \!+\!\sum_{k=1}^{Y}\!a_{km}\!\log_{2}\!\left(\!1\!+\!\!\frac{\!\left|\left(r_{k}\!+\!\mathbf{h}_{km} \!\mathbf{\Theta}_{km}\! \mathbf{g}_{km}\!\right)\rho_k\right|^2}{\sigma^2}\!\right) \right).
\end{equation}
\hrulefill
\end{figure*}

Therefore, the probability that there is a successful transmission as $V_i$ mobile users select the particular channel is $P_{V_i}P_{i,s}$. Hence, the
successful transmission probability in a given channel at the first time can be expressed as
\begin{align}\label{qa}
\setcounter{equation}{15}
\centering
\!P_{1,C}\!&=\!\sum_{V_1=1}^{N_1}P_{V_1}P_{V_1,s}\\ \notag
&=\!\sum_{V_1=1}^{N_1} \left(\!\begin{array}{c} \!N_1 \!\\ \!V_1 \!\end{array}\!\right)\!V_1\!\tau_1\left(1\!-\!\tau_1\right)\!^{\!V_1\!-\!1}\!\left(\!\frac{1}{C} \!\right)^{V_1}\!\left(1\!-\!\frac{1}{C}\!\right)^{\!N_1\!-\!V_1},
\end{align}
where $N_1=Y$. Note that each mobile user selects its best idle channel based on its individual sensing, and then transmits its data to the BS via the RIS on its selected best idle channel. Moreover, by taking the effect of sensing errors into account in the selection of the best idle channel, \eqref{qa} can be rewritten as follows:
\begin{align}\label{qaa}
\centering
\!P_{1,C}\!=\!\sum_{V_1=1}^{Y} \!P_{1,e}\!\left(\!\begin{array}{c} \!Y \!\\ \!V_1 \!\end{array}\!\right)\!V_1\tau_{1}\!\left(1-\tau_{1}\right)\!^{\!V_1\!-\!1}\!\left(\!\frac{1}{C} \!\right)^{V_1}\!\left(1\!-\!\frac{1}{C}\!\right)^{\!Y\!-\!V_1}.
\end{align}

Refer to \cite{Thilina}, the number of mobile users that successfully transmit at the 
first time can be expressed as
\begin{equation}\label{qa1}
\centering
\tilde{N_1}=\left\lfloor CP_{1,C} \right\rfloor.
\end{equation}

According to \eqref{qa1}, the number of mobile users has to contend for their RIS-assisted transmissions at the second time can be denoted by
\begin{equation}\label{qa2}
\centering
N_2=Y-\left\lfloor CP_{1,C}\right\rfloor.
\end{equation}

Then, the number of mobile users that successfully transmit at the first and second times can be expressed as
\begin{equation}\label{qa3}
\centering
\tilde{N_2}=\left\lfloor C \sum_{l=1}^{2}P_{l,C}\right\rfloor,
\end{equation}
where 
\begin{equation}\label{qa4}
\centering
\!P_{2,C}\!=\!\sum_{\!V_2 \!= 1}^{\!N_2}\!P_{2,e}\!\left(\!\begin{array}{c} \!N_2 \!\\ \!V_2 \!\end{array}\!\right)\!V_2\tau_2\!\left(1\!-\!\tau_2 \right)^{\!V_2\!-\!1}\!\left(\!\frac{1}{C} \!\right)^{V_2}\!\left(1\!-\!\frac{1}{C}\!\right)^{N_2\!-\!V_2}.
\end{equation}

Similarly, the successful transmission probability at the $i$th time and the number of mobile users successfully transmit at the $i$th time can be respectively expressed as
\begin{equation}\label{qa5}
\centering
\!P_{i,C}\!=\!\sum_{\!V_i=1}^{N_i}\!P_{i,e}\left(\!\begin{array}{c} \!N_i \!\\ \!V_i \!\end{array}\!\right)V_i\tau_i\!\left(1\!-\!\tau_i\right)^{V_i-1}\!\left(\!\frac{1}{C} \!\right)^{V_i}\!\left(1\!-\!\frac{1}{C}\!\right)^{N_i\!-\!V_i}
\end{equation}
and 
\begin{equation}\label{qa6}
\centering
\tilde{N_i}=\left\lfloor C \sum_{l=1}^{i}P_{l,C}\right\rfloor.
\end{equation}
In \eqref{qa5}, the number of mobile users has to contend for their access at the $i$th time, $N_i$, is denoted by
\begin{equation}\label{qa7}
\centering
N_i=\left\lfloor Y-C \sum_{l=1}^{i-1}P_{l,C}\right\rfloor.
\end{equation}

Therefore, to guarantee the fairness of each mobile user (i.e., each mobile user can transmit one time), the number of contention times that can successfully meet the requirement of $Y$ mobile users, $N_r$, is expressed as
\begin{equation}\label{qa8}
\centering
N_r=\sum_{i}I\left({N_i}\ge 0\right),
\end{equation}
where
\begin{equation} \label{qa9}
I\left({N_i}\ge 0\right) = \begin{cases}
1,    & \text{if}\ {N_i}\ge 0, \\ 
0, & \text{otherwise}.
\end{cases}
\end{equation}

According to \cite{Bianchi}, it is known that the time length required by one successful contention and transmission can be expressed as
\begin{equation} 
t_r\!=\!RTS\!+CTS\!+\!t_{d}\!+\!2{SIFS}\!+\!{DIFS}\!+\!2\delta,
\label{qa10}
\end{equation}
where $t_d$ is the time length of a payload required by the mobile user. Besides, $RTS$, $CTS$, $SIFS$, and $DFIS$ are the duration of request-to-send (RTS), clear-to-send (CTS), short inter-frame space (SISF), and DCF inter-frame space (DISF), respectively.

Based on the above analysis, the system throughout of RIS-assisted contended transmissions, ${\cal S}_c$, is the sum throughput of each contended mobile user, which can be expressed as
\begin{equation}\label{qa11}
\centering 
{\cal S}_c =\frac{1}{\beta t_2}\sum_{k=1}^{Y}\sum_{m=1}^{M}a_{km}t_d\frac{B}{C}\log_{2}\left(1+\text{SNR}_{km}\right),
\end{equation} 
where $\sum_{m=1}^{M}a_{km}=1,\forall k\in {\cal Y}$.

According to \eqref{SNRE1}, \eqref{qa11} can be rewritten as
\begin{equation}\label{qa12}
\centering 
{\cal S}_c \!=\!\frac{Bt_d}{\!C\!\beta  t_2}\sum_{k=1}^{Y}\sum_{m=1}^{M}\!a_{km}\!\log_{2}\!\left(\!1\!+\!\frac{\!\left|\left(r_{k}\!+\!\mathbf{h}_{km}\! \mathbf{\Theta}_{km}\! \mathbf{g}_{km}\!\right)\rho_k\right|^2}{\sigma^2}\!\right),
\end{equation} 

To guarantee the RIS-assisted transmission of each mobile user, the contended transmission period has to meet the following condition
\begin{equation}\label{qa13}
\centering 
\beta t_2 \ge N_rt_r.
\end{equation}

According to \eqref{qa12}, the parameters of the proposed RIS-assisted MAC protocol can be designed as 
\begin{equation}\label{qa14}
\centering 
\beta:\alpha \ge (N_rt_r): (Jt).
\end{equation}

\begin{remark}
Based on the throughput analysis in \eqref{Ss1} and \eqref{qa12}, the overall throughput included the scheduled and contended transmissions can be expressed as \eqref{So1}, where $\sum_{m=1}^{M}a_{km}=1,\forall k\in \cal X$, and $\sum_{m=1}^{M}a_{km}=1,\forall k\in \cal Y$. To improve the performance of the proposed joint multiple access framework, the MAC protocol parameters (e.g., $\cal T$, ${\alpha, \beta}$, and $\rho_k^2$ ) and the RIS configuration parameters (e.g., $a_{km}$ and $\mathbf{\Theta}_{km}$) can be jointly optimized for static and mobile users.
\end{remark}

\section{Joint Optimization and Solution}\label{MA soving}
In this section, we first formulate a joint optimization problem in Section \ref{PF}. We then decompose the original problem in Section \ref{PDS}, and we solve each sub-problem in Section \ref{PMAC} and Section \ref{PRIS}. Finally, we discuss the complexity in Section \ref{CPI}.  

\subsection{Problem Formulation}
\label{PF}
In this paper, we aim for maximizing the overall system throughput by jointly optimizing the MAC protocol and the RIS configuration parameters. Specifically, the proposed joint optimization is formulated as 
\begin{align}
\setcounter{equation}{32}
\mathbb{P}_0: &\underset{\{{\cal T}, {\alpha, \beta}, {t_{kj}}, {\rho_k^2}, a_{km}, {\bm \Psi}\}}{\rm max}\;\; {\cal S}_o \label{problem}  \\ 
{\rm{s}}{\rm{.t}}{\rm{.}} & \;\; t_0\textgreater 0,\;\; t_1\textgreater 0, \;\; t_2\textgreater 0, \tag{\ref{problem}{a}} \label{problema}\\ 
& \;\; \alpha+\beta= 1,\; \frac{\beta}{\alpha} \ge \frac{N_rt_r}{Jt}, \; \alpha\in[0,1],\; \beta\in[0, 1], \tag{\ref{problem}{b}} \label{problemb}\\ 
& \;\; t_{kj} \in \{0,1\}, \;\; \forall k \in {\cal X}, \ \forall j \in {\cal J},  \tag{\ref{problem}{c}} \label{problemc}\\ 
& \;\sum_{j =1}^{J} t_{kj}= 1, \ \forall k \in {\cal X}, \tag{\ref{problem}{d}} \label{problemd}\\
& \;\;  \sum_{k=1}^{X}\rho_k^2 \leq P_{max}, \;\; \forall k \in {\cal X}, \tag{\ref{problem}{e}} \label{probleme}\\ 
& \;\;  \rho_k^2 = \Upsilon, \;\; \forall k \in {\cal Y}, \tag{\ref{problem}{f}} \label{problemf}\\ 
& \;\; \!\frac{B}{C}\!\log\!_{2}\!\left(\!1\!+\!\text{\!SNR}_{\!km}\!\right)\ge R_{min},\ \forall k \in \bar{\cal K}, \forall m \in {\cal M}, \tag{\ref{problem}{g}} \label{problemg}\\ 
& \;\; a_{km}\in \left \{0,1  \right \}, \ \forall k \in \bar{\cal K},  \forall m \in {\cal M}, \tag{\ref{problem}{h}} \label{problemh} \\
& \; \sum_{m =1}^{M} a_{km}= 1,\ \forall k \in {\cal X}, \tag{\ref{problem}{i}} \label{problemi}\\
& \; \sum_{k =1}^{X} a_{km}= J, \ \forall m \in {\cal M}, \tag{\ref{problem}{j}} \label{problemj}\\
& \;\sum_{m =1}^{M} a_{km}= 1, \ \forall k \in {\cal Y}, \tag{\ref{problem}{k}} \label{problemk}\\
& \;\; \vert\gamma_{km}^{n}\vert=1, \ \forall n \in {\cal N}, \ \forall k \in \bar{\cal K}, \ \forall m \in {\cal M}, \tag{\ref{problem}{l}} \label{probleml}  
\\
& \;\; \theta_{km}^n \in [0, 2\pi),  \forall n \in {\cal N}, \ \forall k \in \bar{\cal K}, \ \forall m \in {\cal M},   \tag{\ref{problem}{m}} \label{problemm}  
 \end{align}
where (\ref{problem}a) limits the time length of $t_0, t_1$, and $t_2$. (\ref{problem}b) represents the feasibility of $\alpha$ and $\beta$. (\ref{problem}c) indicates that $t_{kj}$ is a binary value, where $t_{kj}=1$ represents that slot $D_{cj}$ is allocated to $U_k$, otherwise $t_{kj}=0$. (\ref{problem}d) indicates that at most one data slot is used by a static user in a scheduled transmission period. (\ref{problem}e) indicates that the sum transmit power of all static users has to be less than a maximum transmit power, $P_{max}$. (\ref{problem}f) indicates that the transmit power of each mobile user is fixed as $\Upsilon$. (\ref{problem}g) indicates that the data rate of $U_k$ has to be higher than a data rate threshold $R_{min}$. (\ref{problem}h) indicates that $a_{km}$ is a binary value, where $a_{km}=1$ represents that the $R_m$ is used by the $U_k$, and $a_{km}=0$, otherwise. (\ref{problem}i) indicates that at most one RIS is allocated to a static user in a scheduled transmission period. (\ref{problem}j) indicates that one RIS serves $J$ mobile users in a scheduled transmission period. (\ref{problem}k) indicates that at most one RIS is used by a mobile user in a contended transmission period. (\ref{problem}l) and (\ref{problem}m) indicate the feasibility of each RIS element's amplitude and phase shift.
 
We observe that the proposed joint optimization problem $\mathbb{P}_0$ in (\ref{problem}) is an MINLP problem, which is NP-hard and whose globally optimal solution is difficult to obtain by using the common standard optimization approaches. On one hand, since the switching of the scheduled transmissions and the contended transmissions, only the static users have to be scheduled by the BS within $\alpha t_2$, while the mobile users contend for the RISs resources randomly in $\beta t_2$. On the other hand, the RIS configuration of each static user is computed at the BS, while the RIS configuration of each mobile user is computed by itself once it contends for the RIS. Hence, the MAC protocol optimization will significantly affect the computational complexity of the RIS configuration. Moreover, the RIS configuration also significantly affects the overall system throughput of the designed MAC protocol. 

To this end, an alternative optimization method can be invoked as an intuitive approach to solve problem $\mathbb{P}_0$ in (\ref{problem}). Here, we first decompose problem $\mathbb{P}$ into two sub-problems; one is the optimization of the MAC protocol, and the other one is the optimization of the RIS configuration. We then solve them according to the iteration method to achieve the joint design optimization. 

\subsection{Problem Decomposition}
\label{PDS}
By using the Tammer decomposition method, the original problem $\mathbb{P}_0$ is rewritten as
\begin{equation} \label{Pr}
\begin{aligned}
& \hat{\mathbb{P}}_0: \underset{\{ a_{km}, {\bm \Psi} \}}{\rm max}\left(\underset{\{{\cal T}, {\alpha, \beta}, {t_{kj}}, {\rho_k^2}\}}{\rm max}\; {\cal S}_o \right)   \\
&{\rm{s}}{\rm{.t}}{\rm{.}}\;\;\;(\ref{problem}a)-(\ref{problem}m).
 \end{aligned}
\end{equation}

To solve the equivalent problem $\hat{\mathbb{P}}_0$ in \eqref{Pr}, the transformed two sub-problems are illustrated as

\subsubsection{Sub-problem $\mathbb{P}_1$} MAC protocol optimization with fixing the RIS configuration to maximize the overall system throughput, i.e.,
\begin{equation} \label{RA1}
\begin{aligned}
& \mathbb{P}_1: \underset{\{{\cal T}, {\alpha, \beta}, {t_{kj}}, {\rho_k^2}\}}{\rm max}\; {\cal S}_o    \\
&{\rm{s}}{\rm{.t}}{\rm{.}}\;\;\;(\ref{problem}a)-(\ref{problem}g).
 \end{aligned}
\end{equation}

\begin{remark} 
\label{R1}
In \eqref{RA1}, we turn the MAC protocol optimization into further study on the time resource allocation problem and the power resource allocation problem, which can be computed at the BS. Thus, sub-problem $\mathbb{P}_1$ can be transformed as
\begin{equation} \label{Pr1}
\begin{aligned}
& \hat{\mathbb{P}}_1: \underset{\rho_k^2}{\rm max}\left(\underset{\{{\cal T}, {\alpha, \beta}, {t_{kj}}\}}{\rm max}\; {\cal S}_o \right)   \\
&{\rm{s}}{\rm{.t}}{\rm{.}}\;\;\;(\ref{problem}a)-(\ref{problem}g).
 \end{aligned}
\end{equation}
\end{remark}

\subsubsection{Sub-problem $\mathbb{P}_2$} RIS configuration with fixing the MAC protocol parameters to maximize the overall system throughput, i.e.,
\begin{equation} \label{RA2}
\begin{aligned}
& \mathbb{P}_2: \underset{\{a_{km}, {\bm \Psi}\}}{\rm max}\; {\cal S}_o^*    \\
&{\rm{s}}{\rm{.t}}{\rm{.}}\;\;\;(\ref{problem}f)-(\ref{problem}m).
 \end{aligned}
\end{equation}

\begin{remark} 
\label{R2}
In \eqref{RA2}, we turn the RIS configuration into further study on the following two problems: i) the scheduled transmission optimization problem at the BS, and ii) the contended transmission optimization problem at each mobile user. Thus, sub-problem $\mathbb{P}_2$ can be transformed as
\begin{equation} \label{Pr2}
\begin{aligned}
& \hat{\mathbb{P}}_2: \left(\!\underset{\{\!a_{km},  \bm \Psi\}}{\!\rm max}\; {\cal S}_s \right)+\left(\!\underset{\{\!a_{km}, \bm \Psi\}}{\!\rm max}\; {\cal S}_c \right)   \\
&{\rm{s}}{\rm{.t}}{\rm{.}}\;\;\;(\ref{problem}f)-(\ref{problem}m).
 \end{aligned}
\end{equation}
\end{remark}

\subsection{Solution of The MAC Design Problem}
\label{PMAC}
Through observing the objective function and constraints of the problem $\mathbb{P}_0$ in (\ref{problem}), as the usage state of RISs ($a_{km}$) and RIS configuration (${\bm \Psi}$) are fixed, the allocation of data slot (i.e., $t_{kj}$) will not affect the overall system throughput. Thus, solving $\hat{\mathbb{P}}_1$ is equivalent to solving 
the following two problems:
\begin{equation}\label{BA1}
\begin{aligned}
& \mathbb{P}_{1a}: \ \underset{\{{\cal T}, \alpha, \beta\}}{\rm min}\;  
\{t_0+t_1+t_2\}  \\
& {\rm{s}}{\rm{.t}}{\rm{.}}\;\;\;\;\;\;(\ref{problem}a)-(\ref{problem}d),
 \end{aligned}
\end{equation}
and 
\begin{equation}\label{BA2}
\begin{aligned}
&\!\mathbb{P}_{1b}: \ \!\underset{\rho_k^{2}}{\rm max}\; \!\sum_{k=1}^{X}\!\sum_{m=1}^{M}\!a_{km}\!\log_{2}\!\left(\!1\!+\!\frac{\!\left|\left(r_{k}\!+\!\mathbf{h}_{km}\! \mathbf{\Theta}_{km} \!\mathbf{g}_{km}\!\right)\!\rho_k\right|^2}{\sigma^2}\!\right) \\ 
&{\rm{s}}{\rm{.t}}{\rm{.}}\;\;\;\;\;\;(\ref{problem}e),(\ref{problem}g). 
\end{aligned}
\end{equation}

To solve problem $\mathbb{P}_{1a}$, $t_0$, $t_1$, and $t_2$ should be minimized. According to the analysis of $N_r$ and $t_r$ in Section \ref{MA performance}, to guarantee fairness of each user (i.e., each user can use one RIS only one time in a frame), the optimizations of $t_0$, $t_1$, and $t_2$ are given by
\begin{equation}\label{T1}
\centering
\left\{\begin{array}{l}
t_0^*=Kt_p, \\
t_1^*={\cal O}_{min},\\
t_2^*=J^*t+N_rt_r,
\end{array} \right.
\end{equation}
where $t_p$ is the time of one pilot transmission, ${\cal O}_{min}$ is the minimum computational complexity (i.e., the minimum iteration time), and $J^*=\frac{X}{C}$. 

Based on the optimal $t_2^*$, the optimal $\alpha^*$ and $\beta^*$ can be expressed as
\begin{equation}\label{T2}
\centering
\left\{\begin{array}{l}
\alpha^*=\frac{J^*t}{J^*t+N_rt_r},\\
\beta^*=\frac{N_rt_r}{J^*t+N_rt_r}.
\end{array} \right.
\end{equation}
Besides, given $a_{km}$ and ${\bm \Psi}$, problem $\mathbb{P}_{1b}$ can be easily solved using the existed math tools since it is a strictly convex
problem. In the end, the solution of the MAC design problem is illustrated in the proposed Algorithm \ref{alg1}.

\renewcommand{\algorithmicrequire}{\textbf{Initialization:}} 
\renewcommand{\algorithmicensure}{\textbf{Output:}} 
\begin{algorithm}[t]    
\caption{MAC Design Algorithm for Solving $\mathbb{P}_{1}$ at the BS}             
\small
\label{alg1}                  
\begin{algorithmic}[1]             
\REQUIRE $K$, $C$, $M$, $Z$, $t_p$, $t$, $\sigma^2$, ${\bf g}_{km}[0]$, ${\bf h}_{km}[0]$, $\mathbf{\Theta}_{km}[0]$, $a_{km}[0]$, $\rho_k^2[0]$, the maximum iteration number is $L_1$, and set $l_1 = 0$;
\STATE  \textbf{for} user $k\in \cal K$, confirms $u_k$;
\STATE  \textbf{end for}
\STATE  Obtains $X$ and $Y$ according to \eqref{U};
\STATE  Solves the optimal $\cal T^*$ according to \eqref{BA1};
\STATE  Calculates $J^*$ according to \eqref{J};
\STATE  Calculates $N_r$ according to \eqref{qa8};
\STATE  Calculates $t_r$ according to \eqref{qa10};
\STATE  Calculates $\alpha^*$ and $\beta^*$ according to \eqref{T2};
\STATE  \textbf{for} user $k\in \cal X$;
\STATE  \textbf{repeat}
\STATE  With fixing ${\bf g}_{km}[l_1]$, ${\bf h}_{km}[l_1]$, $\mathbf{\Phi}_{km}[l_1]$, $a_{km}[l_1]$, obtains \\ ${\rho_k^2}[l_1]$ according to \eqref{BA2}; 
\STATE  Based on the optimal $\cal T^*$, $\alpha^*$, $\beta^*$, and ${\rho_k^2}[l_1]$, obtains ${\bf g}_{km}[l_1+1]$, ${\bf h}_{km}[l_1+1]$, $\mathbf{\Theta}_{km}[l_1+1]$, $a_{km}[l_1+1]$ \\ according to \eqref{SA1} and \eqref{SA2}; 
\STATE  Updates $l_1 \leftarrow l_1 + 1$;
\STATE  \textbf{until} $l_1\ge L_1$;
\STATE  ${\rho_k^2}^*=\rho_k^2$;
\STATE  \textbf{end for}
\ENSURE ${\cal T}^*,\alpha^*, \beta^*$, and ${\rho_k^2}^*$.\\  
\end{algorithmic}
\end{algorithm}

\subsection{Solution of The RIS Configuration Problem}
\label{PRIS}
For the sub-problem $\mathbb{P}_2$, since the MAC protocol parameters (${\cal  T}, \alpha, \beta, \rho_k^2$) are fixed, $\hat{\mathbb{P}}_2$ in \eqref{Pr2} can be decomposed as the centralized RIS configuration and the distributed RIS configuration at the BS and mobile users, respectively. The solution of each problem is presented in Algorithms \ref{alg2} and \ref{alg3}, respectively.

\subsubsection{Centralized RIS Configuration at The BS}

\begin{equation}\label{SA1}
\begin{aligned}
& \!\mathbb{P}_{\!2a}\!: \!\underset{\{\!a_{km}, \!{\bm \Psi}\}}{\!\rm max}\; \! \sum_{k=1}^{X}\!\sum_{m=1}^{M}\!a_{km}\!\log_{2}\!\left(\!1\!+\!\frac{\!\left|\left(r_{k}\!+\!\mathbf{h}_{km}\! \mathbf{\Theta}_{km} \!\mathbf{g}_{km}\!\right)\!\rho_k\right|^2}{\sigma^2}\!\right)  \\
& {\rm{s}}{\rm{.t}}{\rm{.}}\;\;\;\;\;\;(\ref{problem}g)-(\ref{problem}j),(\ref{problem}l) , (\ref{problem}m).
 \end{aligned}
\end{equation}
Problem $\mathbb{P}_{2a}$ is an MINLP problem. To solve problem $\mathbb{P}_{2a}$, multiple alternative iterations between $a_{km}$ and $\bf \Psi$ are operated at the BS to obtain the RIS allocation and RIS phase shifts of $X$ static users.

\begin{itemize}
\item[(a)] \textit{Centralized RIS allocation optimization}. Fixed the phase shifts ($\bm \Psi$) of $M$ RISs, problem $\mathbb{P}_{2a}$ can be rewritten as
\end{itemize}
\begin{equation}\label{SA11}
\begin{aligned}
& \!\mathbb{P}_{2a_1}\!: \underset{a_{km}}{\!\rm max}\; \! \sum_{k=1}^{X}\!\sum_{m=1}^{M}\!a_{km}\!\log_{2}\!\left(\!1\!+\!\frac{\!\left|\left(r_{k}\!+\!\mathbf{h}_{km}\! \mathbf{\Theta}_{km} \!\mathbf{g}_{km}\!\right)\!\rho_k\right|^2}{\sigma^2}\!\right)  \\
& {\rm{s}}{\rm{.t}}{\rm{.}}\;\;\;\;\;\;(\ref{problem}h) -(\ref{problem}j),
 \end{aligned}
\end{equation}
where problem $\mathbb{P}_{2a_1}$ is a $``0-1"$ linear programming problem, which can be solved by the existed math tool. 

\begin{itemize}
\item[(b)] \textit{Centralized RIS phase shift optimization}. Fixed the RIS allocation ($a_{km}$) of $X$ static users, problem $\mathbb{P}_{2a}$ can be rewritten as
\end{itemize}
\begin{equation}\label{SA12}
\begin{aligned}
& \!\mathbb{P}_{2a_2}\!: \underset{\bf \Psi}{\!\rm max}\; \! \sum_{k=1}^{X}\!\sum_{m=1}^{M}\!a_{km}\!\log_{2}\!\left(\!1\!+\!\frac{\!\left|\left(r_{k}\!+\!\mathbf{h}_{km}\! \mathbf{\Theta}_{km} \!\mathbf{g}_{km}\!\right)\!\rho_k\right|^2}{\sigma^2}\!\right)  \\
& {\rm{s}}{\rm{.t}}{\rm{.}}\;\;\;\;\;\;(\ref{problem}g) , (\ref{problem}l), (\ref{problem}m).
 \end{aligned}
\end{equation}

Problem $\mathbb{P}_{2a_2}$ is non-convexity, and the optimization result of $\mathbb{P}_{2a_2}$ can be solved at the BS, which is highlighted in the following observation. 

\theoremstyle{Observation}
\newtheorem{observation}{\textbf{Observation}}
\begin{observation} \label{O1}
As the RIS allocation ($a_{km}$) is given, the optimal solution to problem $\mathbb{P}_{\!2\!-\!1b}$ is the one that maximizes the channel gain via RISs. With this mind, its optimal solution, $\Psi^*$, is denoted by ${\bf\Psi^*}=\{{{\bm\theta}_{11}}^*,\ldots,{\bm{\theta}_{km}}^*,\ldots,{\bm{\theta}_{KM}}^*\}$, where ${\bm{\theta}_{km}}^*=\{\theta_{km}^{{1}^*}, \ldots, \theta_{km}^{{n}^*}, \ldots \theta_{km}^{{N}^*}\}$, and $\theta_{km}^{{n}^*} = \arg(r_k)-\arg({h}_{km}^{n})-\arg({g}_{km}^{n})$, $ \forall k\in {\cal X},\forall m\in {\cal M}, \forall n\in {\cal N} $.
\end{observation}

\begin{proof}
As for the solution of problem $\mathbb{P}_{\!2\!-\!1b}$, we have the following inequality
\begin{align}\label{ph}
\centering
\vert r_{k} +\mathbf{h}_{km} \mathbf{\Theta}_{km} \mathbf{g}_{km} \vert^2 \le \vert r_k \vert^2+\vert \mathbf{h}_{km} \mathbf{\Theta}_{km} \mathbf{g}_{km}\vert^2 . 
\end{align} 
The equality in \eqref{ph} holds only when the RIS reflection coefficients are equal to $\arg(r_k) \triangleq \arg(\mathbf{h}_{km} \mathbf{\Theta}_{km} \mathbf{g}_{km})$. 

To optimize ${\theta_{km}^{n}}, n \in {\cal N}$, we let $\mathbf{h}_{km} \mathbf{\Theta}_{km} \mathbf{g}_{km}\!=\!\mathbf{w}_{km}\mathbf{\Phi}_{km}$, where $\mathbf{w}_{km}\!=\![w_{km}^{1},\ldots, w_{km}^{n},\ldots, w_{km}^{N}]\!\in\!\mathbb{C}^{1\times(N)}$, $w_k^{n}\!=\!e^{j\theta_{km}^n}, \forall n \in \cal N$, and ${\mathbf{\Phi}}_{km}\!=\!\text{diag}(\mathbf{h}_{km})\mathbf{g}_{km}$. Then, problem $\mathbb{P}_{\!2\!-\!1b}$ can be simplified as 
\begin{align}
& \!\mathbb{P}_{3}\!: \mathop {\max} \limits_{\mathbf{w}_{km}}\ \vert \mathbf{w}_{km} \mathbf{\Phi}_{km} \vert^2\ \ \ \ \ \ \ \ \ \ \ \ \ \ \ \ \ \ \ \ \ \ \ \ \ \ \ \label{problemm} \\ \notag
& {\rm{s}}{\rm{.t}}{\rm{.}}\;  \vert {w_{km}}^{n} \vert^2=1, \ \ \forall k\in {\cal X},\forall m\in {\cal M}, \forall n\in {\cal N}.
\end{align}
It is observed that the optimal phase shift of RIS element $n$ on the $R_m$ for $U_k$ can be obtained by setting $\mathbf{w}_{km}^*=e^{j({\arg(r_k)-\arg(\text{diag}(\mathbf{h}_{km})\mathbf{g}_{km}))}}$, and then we have  
\begin{align}\label{p22bb}
\theta_{km}^{{n}^*} = \arg(r_k)-\arg({h}_{km}^{n})-\arg({g}_{km}^{n}),
\end{align}
where $h_{km}^{n}\in \mathbf {h_k}$, $g_{km}^{n}\in \mathbf {g_k}$,  $\forall k\in {\cal X}, \forall m\in {\cal M}, \forall n\in {\cal N}$. 
\end{proof}

\subsubsection{Distributed RIS Configuration at Mobile Users}
 
\begin{equation}\label{SA2}
\begin{aligned}
& \!\mathbb{P}_{2b}\!: \!\underset{\{\!a_{km}, {\bf \theta}_{km} \}}{\!\rm max}\; \! \!a_{km}\!\log_{2}\!\left(\!1\!+\!\frac{\!\left|\left(r_{k}\!+\!\mathbf{h}_{km}\! \mathbf{\Theta}_{km} \!\mathbf{g}_{km}\!\right)\!\rho_k\right|^2}{\sigma^2}\!\right) \\
& {\rm{s}}{\rm{.t}}{\rm{.}}\;\;\;\;\;\;(\ref{problem}g),(\ref{problem}h),(\ref{problem}k)-(\ref{problem}m).
 \end{aligned}
\end{equation}

To solve problem $\mathbb{P}_{2b}$ at each mobile user, the distributed RIS allocation and RIS phase shift problems are solved using an alternative method.
\begin{itemize}
\item[(a)]\textit{Distributed RIS allocation optimization}. Since each mobile user contends for its RIS, fixed the RIS phase shift ($\bm \theta_{km}$), problem $\mathbb{P}_{2b}$ can be rewritten as
\end{itemize}
\begin{equation}\label{SA21}
\begin{aligned}
& \!\mathbb{P}_{2b_1}\!:  \underset{ a_{km} }{ \rm max}\;  a_{km}\!\log_{2}\!\left(\!1\!+\!\frac{\!\left|\left(r_{k}\!+\!\mathbf{h}_{km}\! \mathbf{\Theta}_{km} \!\mathbf{g}_{km}\!\right)\!\rho_k\right|^2}{\sigma^2}\!\right)  \\
& {\rm{s}}{\rm{.t}}{\rm{.}}\;\;\;\;\;\;(\ref{problem}h),(\ref{problem}k),
 \end{aligned}
\end{equation}
where $\mathbb{P}_{2b_1}$ can be easily solved by finding the best RIS link for $U_k$. 
 
\begin{itemize}
\item[(b)]\textit{Distributed RIS phase shift optimization}. Fixed the RIS allocation ($a_{km}$), problem $\mathbb{P}_{2b}$ can be rewritten as
\end{itemize}
\begin{equation}\label{SA22}
\begin{aligned}
& \!\mathbb{P}_{2b_2}\!:  \underset{ {\bf \theta}_{km} }{ \rm max}\;  a_{km}\!\log_{2}\!\left(\!1\!+\!\frac{\!\left|\left(r_{k}\!+\!\mathbf{h}_{km}\! \mathbf{\Theta}_{km} \!\mathbf{g}_{km}\!\right)\!\rho_k\right|^2}{\sigma^2}\!\right)  \\
& {\rm{s}}{\rm{.t}}{\rm{.}}\;\;\;\;\;\;(\ref{problem}g) , (\ref{problem}l), (\ref{problem}m).
 \end{aligned}
\end{equation}

\renewcommand{\algorithmicrequire}{\textbf{Initialization:}} 
\renewcommand{\algorithmicensure}{\textbf{Output:}} 
\begin{algorithm}[t]    
\caption{Centralized RIS Configuration Algorithm for Solving $\mathbb{P}_{2a}$ at the BS}             
\small
\label{alg2}                  
\begin{algorithmic}[1]             
\REQUIRE $X$, $C$, $M$, $\sigma^2$, ${\bf g}_{km}[0]$, ${\bf h}_{km}[0]$, $\mathbf{\Theta}_{km}[0]$, $a_{km}[0]$, $\rho_k^2[0]$, the maximum iteration number is $L_2$, and set $l_2 = 0$;
\STATE  \textbf{repeat}
\STATE  With given $\cal X$, $\cal T^*$, $\alpha^*$, ${\rho_k^2}^*$, $a_{km}[l_2]$, ${\bf g}_{km}[l_2]$, and ${\bf h}_{km}[l_2]$, obtains $\mathbf{\Psi}[l_2]$ at the BS according to \eqref{SA12}, where $k\in \cal X$;
\STATE  Based on $\mathbf{\Psi}[l_2]$, obtains $a_{km}[l_2+1]$, ${\bf g}_{km}[l_2+1]$, and ${\bf h}_{km}[l_2+1]$ according to \eqref{SA11}, where $k\in \cal X$; 
\STATE  Updates $l_2 \leftarrow l_2 + 1$;
\STATE  \textbf{until} $l_2 \ge L_2$;
\STATE  $a_{km}^*=a_{km}$, $\mathbf{\Psi}^*=\mathbf{\Psi}$;
\ENSURE $a_{km}^*$ and $\mathbf{\Psi}^*, \forall k\in \cal X$.\\  
\end{algorithmic}
\end{algorithm}

\begin{algorithm}[t]    
\caption{Distributed RIS Configuration Algorithm for Solving $\mathbb{P}_{2b}$ at mobile users}             
\small
\label{alg3}                  
\begin{algorithmic}[1]             
\REQUIRE $ Y$, $ C$, $ M$, $ \sigma^2$, $ {\bf g}_{km}[0]$, $ {\bf h}_{km}[0]$, $\!\mathbf{\Theta}_{km}[0]$, $a_{km}[0]$, $\rho_k^2[0]$, the maximum iteration number is $L_3$, and set $l_3 = 0$;
\STATE  \textbf{repeat}
\STATE  With given $\cal Y$, $\cal T^*$, $\beta^*$, ${\rho_k^2}^*$, $a_{km}[l_3]$, ${\bf g}_{km}[l_3]$, and ${\bf h}_{km}[l_3]$, obtains $\mathbf{\theta}_{km}[l_3]$ at $U_k$ according to \eqref{SA22}, where $k\in \cal Y$; 
\STATE  Based on $\mathbf{\theta}_{km}[l_3]$, obtains $a_{km}[l_3+1]$, ${\bf g}_{km}[l_3+1]$, and ${\bf h}_{km}[l_3+1]$ according to \eqref{SA21}, where $k\in \cal Y$; 
\STATE  Updates $l_3 \leftarrow l_3 + 1$;
\STATE  \textbf{until} $l_3 \ge L_3$;
\STATE  $a_{km}^*=a_{km}$, $\mathbf{\theta}_{km}^*=\mathbf{\theta}_{km}$;
\ENSURE $a_{km}^*$ and $\mathbf{\theta}_{km}^*, \forall k\in \cal Y$.\\  
\end{algorithmic}
\end{algorithm}

Problem $\mathbb{P}_{2b_2}$ is non-convexity. The optimization result of $\mathbb{P}_{\!2b_2}$ can be solved at each mobile user, and some comments are highlighted in the following observation. 

\begin{observation} \label{O2}
As the RIS, $R_m$, is obtained by $U_k$, the optimal solution of problem $\mathbb{P}_{\!2\!-\!2b}$ is the one that maximizes the channel gain of $U_k$ via $R_m$. Therefore, its optimal solution, ${\bm{\theta}_{km}}^*$, is denoted by ${\bm{\theta}_{km}}^*=\{\theta_{km}^{{1}^*}, \ldots, \theta_{km}^{{n}^*}, \ldots \theta_{km}^{{N}^*}\}$, where $\theta_{km}^{{n}^*} = \arg(r_k)-\arg({h}_{km}^{n})-\arg({g}_{km}^{n})$, $ \forall k\in {\cal Y},\forall m\in {\cal M}, \forall n\in {\cal N} $.
\end{observation}
\begin{proof}
The proof of \textbf{Observation \ref{O2}} refers to \textbf{Observation \ref{O1}}.
\end{proof}

\subsection{Complexity and Performance Improvement}
\label{CPI}
\subsubsection{Computational Complexity}
To solve problem $\mathbb{P}_0$ in (\ref{problem}), the major involved complexity includes solving the MAC design problem, the centralized RIS configuration problem, and the distributed RIS configuration problem. 

\begin{itemize}
\item  Complexity for solving the MAC design optimization problem: To get the optimal MAC protocol parameters, the complexity is decided by the computing of ${\cal T}^*, \alpha^*, \beta^*$, and $\rho_k^* $. The involved computational complexity in Algorithm \ref{alg1} is ${\cal O}(K+X^3L_1+M^2N^2L_1+X^2M^2L_1)$.

\item  Complexity for solving the centralized RIS configuration problem: By using the alternative iteration method, the involved complexity is decided by the computing of the RIS allocation and the RIS phase shifts of $X$ static users at the BS. Therefore, the involved computational complexity at the BS in Algorithm \ref{alg2} is ${\cal O}(M^2N^2L_2+X^2M^2L_2)$.

\item  Complexity for solving the distributed RIS configuration problem: By using the alternative iteration method, the involved complexity is decided by the computing of the distributed RIS allocation and RIS phase shifts at a mobile user. Therefore, the involved computational complexity at a mobile user in Algorithm \ref{alg3} is ${\cal O}(\hat M^2N^2L_3+\hat M^2L_3)$, where $\hat M$ is the number of unused RISs.
\end{itemize}

It can be seen that the computational complexity at the BS can be reduced since the RIS configuration for static and mobile users is optimized by the BS and each mobile user, respectively.

\subsubsection{MAC Protocol Performance Improvement by RISs}
In the proposed multiple access framework, the benefit of RISs contains two folds: i) RISs assist each user to improve the data rate, and ii) the different RISs design for static and mobile users can improve channel efficiency due to the low computational complexity. Based on these, the MAC protocol performance improvement due to RISs is presented in the following observation. 

\begin{observation}
\label{Tm0}
For the proposed multiple access framework, denote $\Delta \cal O$ as the decrement of computational complexity. The performance improvement brought by decreasing RIS configuration complexity is given by $\frac{T + \Delta \cal O}{T}$, and the improvement brought by designing RIS configuration is limited by $\sum_{k=1}^{X}\sum_{m=1}^{M}a_{km}\xi_{km}+\sum_{k=1}^{Y}\sum_{m=1}^{M}a_{km}\chi_{km}$, where $\xi_{km}, \forall k \in {\cal X}, \forall m \in {\cal M}$ and $\chi_{km}, \forall k \in {\cal Y}, \forall m \in {\cal M}$ are the data rate increment of $U_k, k\in \bar{\cal K}$ benefited from $R_m$.
\end{observation}

In \textbf{Observation~\ref{Tm0}}, $T$ and $\Delta \cal O$ are expressed as
\begin{equation}
\label{cc}
\centering
\left\{\begin{array}{l}
T = t_0+t_1+t_2, \\
\Delta {\cal O} \approx {\cal O}\left(K^3L_1\right)-{\cal O}\left(X^3L_1\right).
\end{array} \right.
\end{equation}
Since the computational complexity at the BS in the proposed MAC framework is mainly affected by the number of static users, i.e., the computational complexity without considering mobility profiles is affected by ${\cal O}(K^3L_1)$, and that considering mobility profiles is affected by ${\cal O}(X^3L_1)$, the decrement of computational complexity, $\Delta \cal O$, can be obtained in \eqref{cc}.

Moreover, $\xi_{km}$ and $\chi_{km}$ are calculated as
\begin{align} 
\!\{\xi\!_{km},\chi\!_{km}\!\} \!&=\! \!B\left(\!\log_2\!\left(1\!+\!\text{SNR}\!_{km}\right)\!-\log_2\!\left(\!1\!+\!\frac{\!\left|r\!_{k} \rho_k\!\right|^2}{\sigma^2}\!\right)\right)\\ \notag
&=\!B \!\log_2\left( \frac{\kappa+\Delta\kappa}{\kappa}\right),
\end{align}
where $\kappa=\sigma^2+r_k^2\rho_k^2$, and $\Delta\kappa=\left |\mathbf{h}_{km}\! \mathbf{\Theta}_{km} \!\mathbf{g}_{km} \right |^2\rho_k^2  \!+2r_k\left |\mathbf{h}_{km}\! \mathbf{\Theta}_{km} \!\mathbf{g}_{km} \right|\rho_k^2 $.

\section{Simulation Resutls}\label{result}

\subsection{Simulation Settings}

\subsubsection{Network Scenario} We consider a network scenario that consists of a BS, 2 RISs having 128 RIS-elements each, and 200 users (includes 100 static users, 80 mobile users, and 20 new users, the ratio of different types of users is denoted by $\epsilon=5:4:1$). All the users are uniformly distributed in a square area of size $50\times50$ (in meters) with the BS located at ($0$, $0$, $100$) in three-dimensional Cartesian coordinates. From a practical implementational perspective, the location of RISs is generally fixed. Hence the location of RISs is given by ($25$, $50$, $50$), and ($50$, $25$, $50$) is illustrated in Fig. \ref{simulation_set}. Unless stated otherwise, the other simulation parameters are set as follows. A Rician fading channel model is assumed, where the user-RIS and RIS-BS channels benefit from the existence of LoS links having a path loss exponent of 2.2, while the user-BS channels are NLoS links with a path loss exponent of 3.6. The power dissipated at each user is 10 dBm, the noise power is -94 dBm, the number of sub-channels is 2, and the bandwidth of each sub-channel is $B = 10$ MHz. The payload size of each user is $500$ bytes. The packet size of RTS and CTS is set to $24$ bytes and $16$ bytes, respectively. SIFS and DIFS are set to $10$ $\mu$s and $50$ $\mu$s, respectively. The minimum contention window $W_0$ is set to $15$, and the maximum backoff stage $m$ is set to $6$. Furthermore, we assume that each RIS occupies a single sub-channel as a benefit of interference cancellation, and each user is only allowed to use a single RIS to communicate with the BS at a time. Lastly, when the number of users and that of RISs changes, the results have to be evaluated on a frame-by-frame basis.

\begin{figure}[t]
  \centering
  \captionsetup{font={small }}
  \includegraphics[width=2.5in, height=2in]{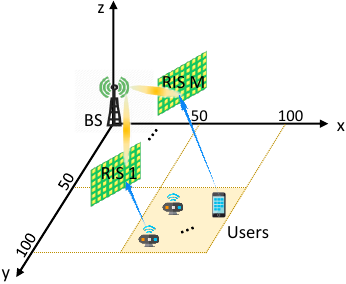}
  \caption{Simulation setup of the proposed multiple access framework.}
  \label{simulation_set}
\end{figure}

\subsubsection{Benchmark Schemes} We consider the following benchmark MAC schemes in the results for comparison.

\begin{itemize}
\item  The centralized multiple access scheme (Scheme 1): The BS schedules the resources and configures RISs for users included the static and the mobile users without contention.

\item  The distributed multiple access scheme (Scheme 2): The static and the mobile users contend for their resources and configures RISs by themselves without scheduling.

\end{itemize}

\subsection{Performance Evaluation}
Figure \ref{S1} evaluates the system throughput of three types of RIS-assisted MAC schemes as the number of users increases. Firstly, it is observed that the system throughput of each scheme is improved with the assistance of RISs. This is because the RIS can help user increase its data rate by improving its link performance. Fig. \ref{S1} shows that the system throughput of Scheme 1 and that of the proposed scheme initially increase, but then tend to saturate as the number of users increases. This is because the computation time ratio of each frame is reduced upon increasing the length of each frame. Then, the system throughput of Scheme 2 exhibits a slight lesion after saturation due to the collisions. Additionally, as the number of users increases, the system throughput of the proposed scheme becomes higher than that of Scheme 1 and Scheme 2. This is because the computational complexity of Scheme 1 and the collision of Scheme 2 increase as the number of users increases. However, as the number of users increases, these two factors can be significantly alleviated in the proposed scheme since the different MAC protocols and RISs configurations are used for the static and mobile users. 

\begin{figure}[t]
  \centering
  \captionsetup{font={small }}
  \includegraphics[width=3.3in, height=2.3in]{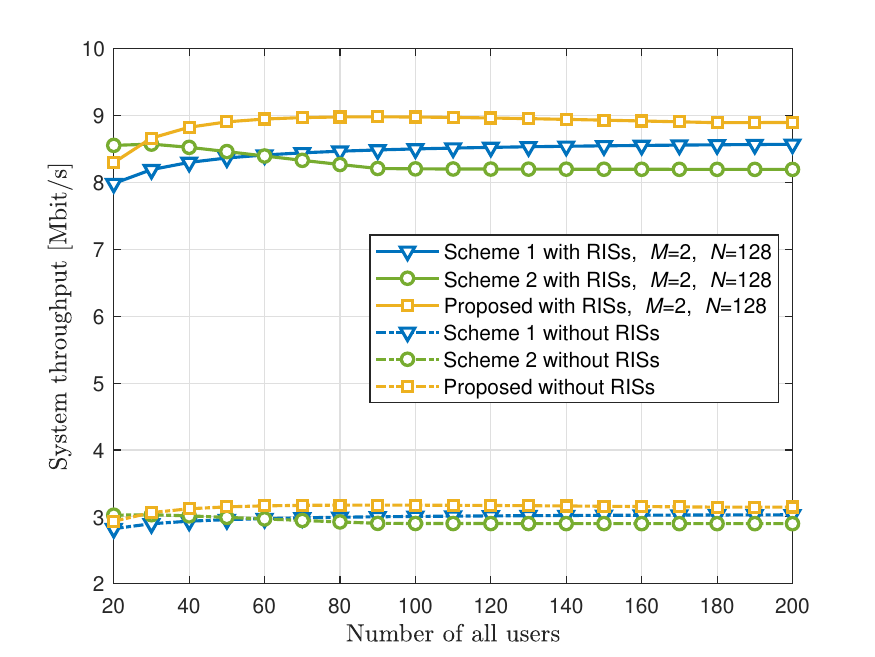}
  \caption{System throughput comparison among the three schemes with RISs or not.}
  \label{S1}
\end{figure}

\begin{figure}[t]
  \centering
  \captionsetup{font={small }}
  \includegraphics[width=3.3in, height=2.3in]{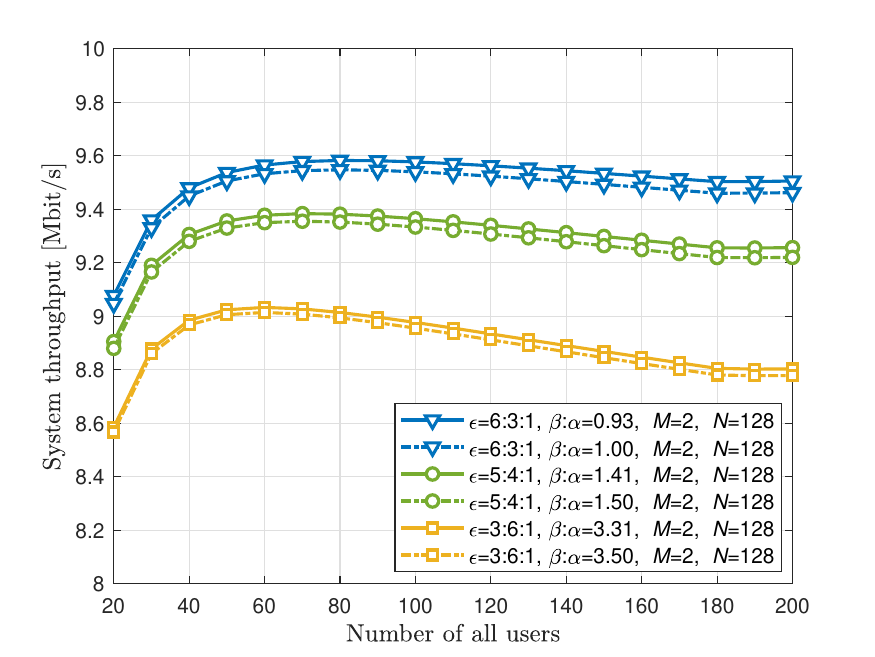}
  \caption{Impact of $\epsilon$ on optimal $\beta/\alpha$ on the system throughput of the proposed scheme.}
  \label{S4}
\end{figure}

\begin{figure}[t]
  \centering
  \captionsetup{font={small }}
  \includegraphics[width=3.3in, height=2.3in]{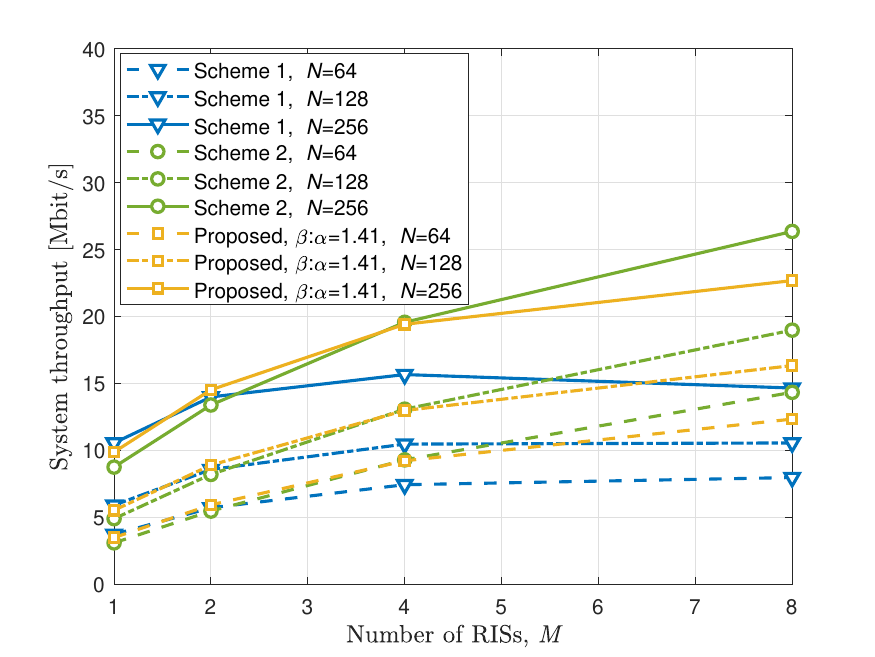}
  \caption{Impact of $N$ and $M$ on the system throughput of the three schemes with $\epsilon\!=\!5\!:\!4\!:\!1$, where the number of users is 200.}
  \label{S12}
\end{figure}

\begin{figure}[t]
  \centering
  \captionsetup{font={small }}
  \includegraphics[width=3.3in, height=2.3in]{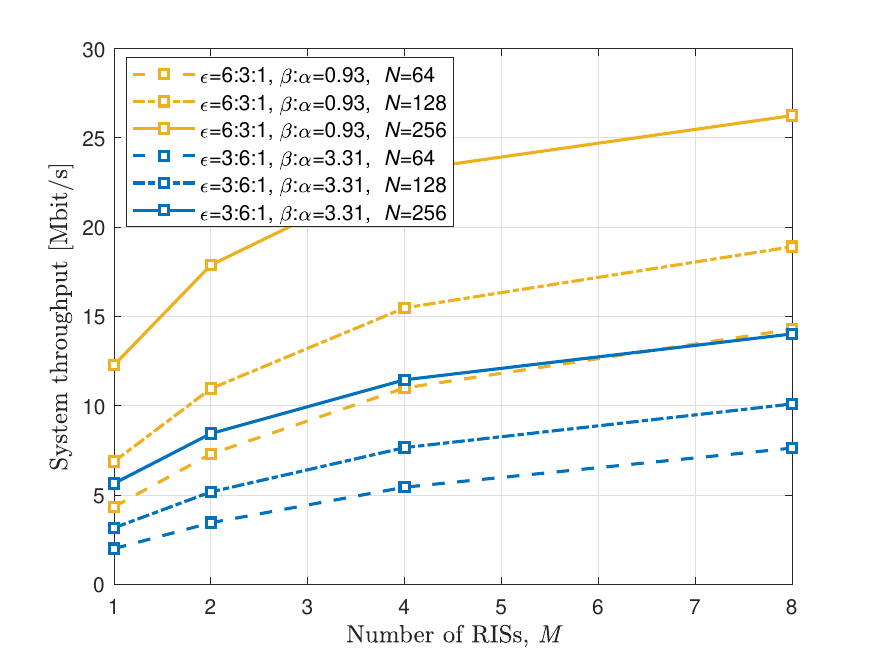}
  \caption{Impact of $\epsilon, N$, and $M$ on the system throughput of the proposed scheme, where the number of users is 200.}
  \label{S11}
\end{figure}

As the number of users increases, Fig. \ref{S4} evaluates the system throughput of the proposed scheme in terms of the MAC protocol parameters. As shown in Fig. \ref{S4},  when the ratio of different types of users, $\epsilon$, is set as $(6\!:\!3\!:\!1), (5\!:\!4\!:\!1)$, and $(3\!:\!6\!:\!1)$, the system throughput of the proposed scheme decreases accordingly. This is because the ratio of mobile users is increased and that of static users is decreased, which leads to a bigger $\beta/\alpha$. In the proposed scheme, as $\beta/\alpha$ increases, the system throughput of the proposed scheme decreases, and the decrement of system throughput gradually becomes obvious. Here, the bigger $\beta/\alpha$ means that the longer time will be spent on the contended transmissions. However, the longer time also serves the same number of mobile users and transmits the same payloads. In other words, the access efficiency of the proposed scheme decreases as $\beta/\alpha$ increases. This is mainly because the channel resources are not fully used during the contended transmission period. Additionally, given $\epsilon$, the scheduled transmission period is fixed once the number of static users is given, it is therefore existing an optimal $\beta/\alpha$ that maximizes the system throughput of the proposed scheme, while guaranteeing the fairness of users, e.g., the optimal $\beta/\alpha$ is $0.93, 1.41$ and $3.31$ for each setting of $\epsilon$. Also, for each setting of $\epsilon$, the system throughput has a minor decline as the number of users increases. 

Figure \ref{S12} evaluates the system's throughput of the three schemes in terms of the number of RIS elements and that of RISs, respectively. Observe in Fig. \ref{S12} that the system throughput of each scheme first increases and then saturates as the number of RISs increases. This is because the access efficiency of each scheme can be enhanced by having more RISs. However, more RISs will increase the computational complexity. Give $\epsilon\ (5\!:\!4\!:\!1)$, compared to Scheme 1 and Scheme 2, the proposed scheme performs better when the number of RISs is less than 4. This is because the proposed scheme decreases the computational complexity of static users, while also decreases the collision of mobile users. As the number of RISs is greater than 4, the throughput of Scheme 2 performs best due to the distributed computation and the less collision. Additionally, the throughput of each scheme increases as the number of RIS elements increases due to the enhanced link performance.

Figure \ref{S11} evaluates the system throughput of the proposed scheme in terms of $\epsilon$, the number of RIS elements, and the number of RISs, respectively. Fig. \ref{S11} shows that the system throughput of the proposed scheme increases as the number of RISs increases for each setting of $\epsilon$. This is because the consumed time handling the same traffic payload of all users decreases. Also, as the number of RISs increases, the increment of the system throughput gradually decreases. This is because more time has to been spent on the RIS configuration. Besides, compared to $\epsilon\ (3\!:\!6\!:\!1)$, a better system throughput can be obtained in the setting of $\epsilon\ (6\!:\!3\!:\!1)$. This is because the more static users are scheduled with a low computational cost and the fewer mobile users decrease the collision. At last, for each setting of $\epsilon$, the system throughput of the proposed scheme increases as the number of RIS elements increases.
\begin{figure}[t]
  \centering
  \captionsetup{font={small }}
  \includegraphics[width=3.3in, height=2.3in]{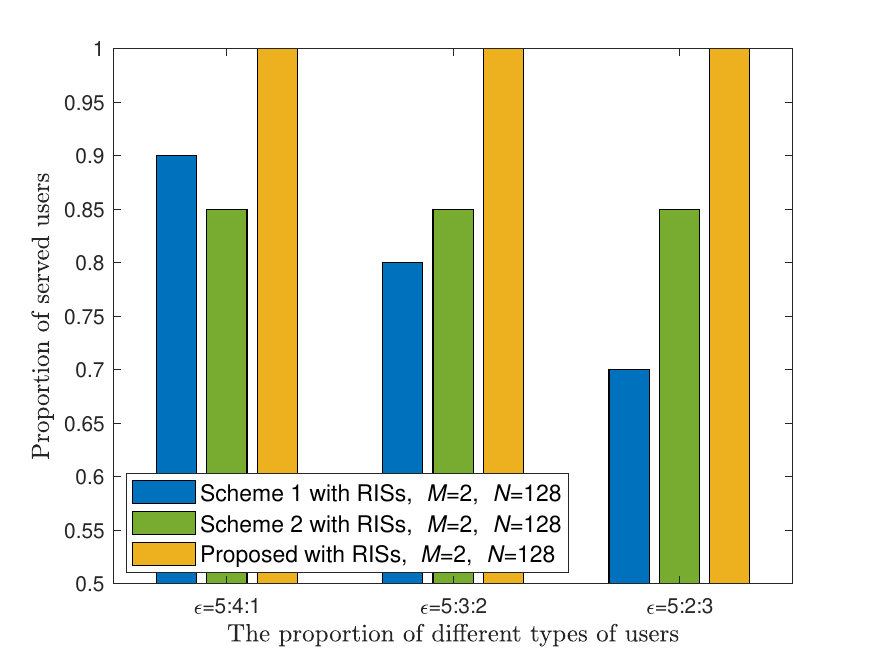}
  \caption{Impact of $\epsilon$ on the fairness of the three schemes, where the optimal $\beta/\alpha$ is selected for the proposed scheme based on each $\epsilon$, and the number of users is 200.}
  \label{S7}
\end{figure}

Figure \ref{S7} evaluates the proportion of served users via RISs in three schemes when $\epsilon$ is set as $(5\!:\!4\!:\!1)$, $(5\!:\!3\!:\!2)$, and $(5\!:\!2\!:\!3)$, respectively. It is observed from Fig. \ref{S7} that the proportion of served users gradually decreases in Scheme 1, while keeps unchanged in the proposed scheme and Scheme 2 when the setting of $\epsilon$ varies from $(5\!:\!4\!:\!1)$ to $(5\!:\!2\!:\!3)$. Compared to Schemes 1 and 2, the proposed scheme serves all users in each case. To be specific, the proposed scheme performs best, and Scheme 2 outperforms Scheme 1 when more new mobile users join in, which can be explained by two factors: 1) the new mobile users always cannot be served in the current frame by RISs in Scheme 1 due to its centralized scheduling; 2) Scheme 2 cannot serve all users since it brings more overheads compared to Scheme 1 in the same transmission time. In Fig. \ref{S8}, the proportion of served users via RISs with the different settings of $\epsilon$ and $\beta/\alpha$ is observed. When the setting of $\epsilon$ is $(6\!:\!3\!:\!1)$, $(5\!:\!4\!:\!1)$, and $(3\!:\!6\!:\!1)$, the proportion of served users in setting of $\epsilon$ is highest at the optimal $\beta/\alpha$ valued $0.93$, $1.41$, and $3.31$. Additionally, if $\beta/\alpha$ is larger than its optimal value in each setting of $\epsilon$, the fairness of each user can be achieved, otherwise, the fairness cannot be guaranteed. Combined with the results in Fig. \ref{S4}, we conclude that there is a trade-off between the system throughput and fairness in the proposed scheme, which is achieved by calculating the optimal $\beta/\alpha$ for a specific $\epsilon$.
\begin{figure}[t]
  \centering
  \captionsetup{font={small }}
  \includegraphics[width=3.3in, height=2.3in]{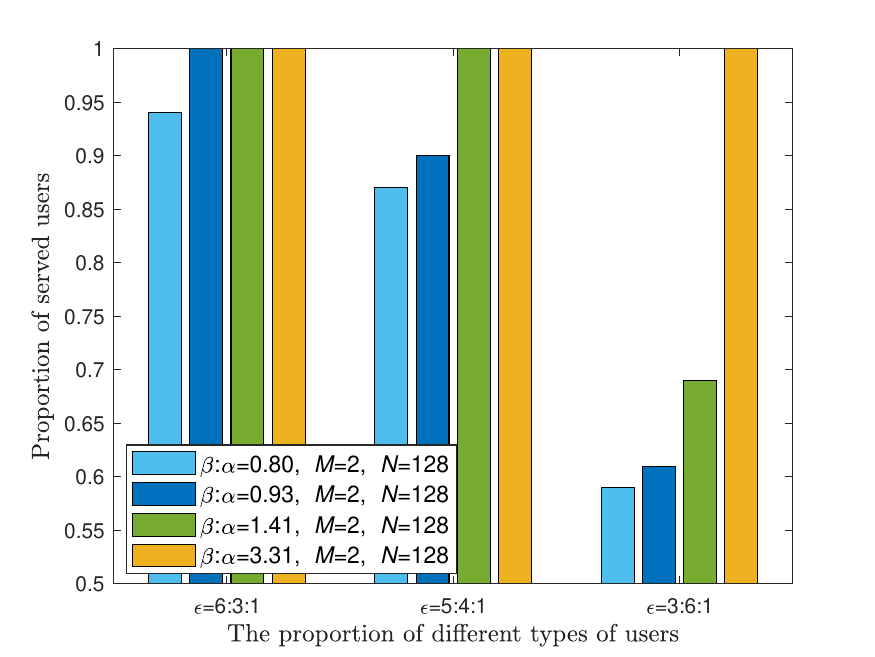}
  \caption{Impact of $\epsilon$ and $\beta/\alpha$ on the fairness of the proposed scheme, where the number of users is 200.}
  \label{S8}
\end{figure}

\section{Conclusion}\label{conclusion}

In this paper, we have conceived a multiple access framework for RIS-assisted communications by integrating the MAC protocol and the RIS configuration into a unified framework. The proposed framework improves the MAC efficiency, while reducing the RISs' computational complexity and offering fairness for multiple users. Our MAC protocol allows different types of users to be assigned to RISs and channel resources by scheduling or contention schemes. By exploiting the interplay between the MAC protocol and the RIS configuration, we have investigated the joint optimization problem of both designs to maximize the overall throughput, while guaranteeing fairness for the users. Then, we have adopted the popular alternative iteration method to obtain the problem's solution. Finally, simulation results have been presented to evaluate the MAC protocol in terms of its throughput and fairness. We have seen that, compared to the benchmarks, the proposed scheme achieves better access fairness and system throughput. As the number of RISs and users increase, harnessing RISs in our multi-user communications scenario substantially expands the resource allocation search space, hence resulting in increased computing complexity and latency. In this context, an AI-based method that shifts the complexity of online computation to offline training may be a potential technique for further addressing the complexity and latency issues of our RIS configuration and MAC protocol \cite{ZZhang, GCA1,add2}.

%


\begin{thebibliography}{100}
\bibitem{WSaad}
W. Saad, M. Bennis, and M. Chen, ``A Vision of 6G Wireless Systems: Applications, Trends, Technologies, and Open Research Problems," \emph{IEEE Network}, vol. 34, no. 3, pp. 134-142, May 2020.
\bibitem{KBLetaief}
K. B. Letaief, W. Chen, Y. Shi, J. Zhang, and Y. A. Zhang, ``The Roadmap to 6G: AI Empowered Wireless Networks," \emph{IEEE Communications Magazine}, vol. 57, no. 8, pp. 84-90, Aug. 2019.
\bibitem{NKato}
N. Kato, B. Mao, F. Tang, Y. Kawamoto, and J. Liu, ``Ten Challenges in Advancing Machine Learning Technologies toward 6G," \emph{IEEE Wireless Communications}, vol. 27, no. 3, pp. 96-103, Jun. 2020.
\bibitem{YLiu}
Y. Liu,  W. Yi, Z. Ding, X. Liu, O. Dobre, and N. Al-Dhahir, ``Application of NOMA in 6G Networks: Future Vision and Research Opportunities for Next Generation Multiple Access," arXiv preprint arXiv:2103.02334, 2021.
\bibitem{XLiu} 
X. Liu, Y. Liu, Y. Chen, and H. V. Poor, ``RIS Enhanced Massive Non-Orthogonal Multiple Access Networks: Deployment and Passive Beamforming Design," \emph{IEEE Journal on Selected Areas in Communications}, vol. 39, no. 4, pp. 1057-1071, Apr. 2021.
\bibitem{XCao} 
X. Cao, B. Yang, C. Huang, C. Yuen, M. Di Renzo, Z. Han, D. Niyao, H. V. Poor, and L. Hanzo, ``AI-Assisted MAC for Reconfigurable Intelligent Surface-Aided Wireless Networks: Challenges and Opportunities,'' \emph{IEEE Communications Magazine}, vol. 59, no. 6, pp. 21-27, Jun. 2021.
\bibitem{YChen} 
Y. Chen \emph{et~al.}, ``Toward the Standardization of Non-Orthogonal Multiple Access for Next Generation Wireless Networks,'' \emph{IEEE Communications Magazine}, vol. 56, no. 3, pp. 19-27, Mar. 2018.
\bibitem{MVaezi}
M. Vaezi, G. A. Aruma Baduge, Y. Liu, A. Arafa, F. Fang, and Z. Ding, ``Interplay Between NOMA and Other Emerging Technologies: A Survey," \emph{IEEE Transactions on Cognitive Communications and Networking}, vol. 5, no. 4, pp. 900-919, Dec. 2019.
\bibitem{YLiu1}
Y. Liu, Z. Qin, M. Elkashlan, Z. Ding, A. Nallanathan, and L. Hanzo, ``Non-Orthogonal Multiple Access for 5G and Beyond,'' \emph{Proceedings of the IEEE}, vol. 105, no. 12, pp. 2347-2381, Dec. 2017.
\bibitem{ZZhang}
Z. Zhang, Y. Li, C. Huang, Q. Guo, C. Yuen and Y. L. Guan, ``DNN-Aided Block Sparse Bayesian Learning for User Activity Detection and Channel Estimation in Grant-Free Non-Orthogonal Random Access,'' \emph{IEEE Transactions on Vehicular Technology}, vol. 68, no. 12, pp. 12000-12012, Dec. 2019.
\bibitem{Liaskos} 
C. Liaskos, A. Tsioliaridou, A. Pitsillides, S. Ioannidis, and I. F. Akyildiz, ``Using any Surface to Realize a New Paradigm for Wireless Communications,'' \emph{Communications of the ACM}, vol. 61, no. 11, pp. 30-33, Nov. 2018.
\bibitem{SHu} 
S. Hu, F. Rusek, and O. Edfors, ``Beyond Massive MIMO: The Potential of Positioning with Large Intelligent Surfaces,'' \emph{IEEE Transactions on Signal Processing}, vol. 66, no. 7, pp. 1761-1774, Apr. 2018.
\bibitem{di2019smart}
M. Di Renzo, M. Debbah, D.-T. Phan-Huy, A. Zappone, M.-S. Alouini, C. Yuen, V. Sciancalepore, G. C. Alexandropoulos, J. Hoydis, H. Gacanin, J. de Rosny, A. Bounceur, G. Lerosey, and M. Fink, ``Smart Radio Environments Empowered by Reconfigurable AI Meta-Surfaces: An Idea Whose Time Has Come,'' \emph{EURASIP Journal on Wireless Communications and Networking}, vol. 2019, no. 1, pp. 1-20, Dec. 2019.
\bibitem{GCA0}  
G. C. Alexandropoulos, G. Lerosey, M. Debbah, and M. Fink, ``Reconfigurable Intelligent Surfaces and Metamaterials: The Potential of Wave Propagation Control for 6G Wireless Communications,'' \emph{IEEE ComSoc TCCN Newsletter}, vol. 6, no. 1, Jun. 2020.  
\bibitem{Renzo} 
B. Yang, X. Cao, C. Huang, Y. L. Guan, C. Yuen, M. Di Renzo, D. Niyato, M. Debbah, L. Hanzo, ``Spectrum Learning-Enabled Reconfigurable Intelligent Surface for Future Green 6G Networks,''  arXiv preprint arXiv:2109.01287, 2021.
\bibitem{GCA00}
G. C. Alexandropoulos, N. Shlezinger, and P. del Hougne, ``Reconfigurable Intelligent Surfaces for Rich Scattering Wireless Communications: Recent Experiments, Challenges, and Opportunities,'' \emph{IEEE Communications Magazine}, vol. 59, no. 6, pp. 28–34, Jun. 2021.
\bibitem{Huang2020}
C.~Huang, S.~Hu, G.~C. Alexandropoulos, A.~Zappone, C.~Yuen, R.~Zhang, M.~Di~Renzo, and M.~Debbah, ``Holographic MIMO Surfaces for 6G Wireless Networks: Opportunities, Challenges, and Trends,'' \emph{IEEE Wireless Communications}, vol. 27, no. 5, pp. 118-125, Oct. 2020.    
\bibitem{MA}
M. A. ElMossallamy, H. Zhang, L. Song, K. G. Seddik, Z. Han, and G. Y. Li, ``Reconfigurable Intelligent Surfaces for Wireless Communications: Principles, Challenges, and Opportunities,'' \emph{IEEE Transactions on Cognitive Communications and Networking}, vol. 6, no. 3, pp. 990-1002, Sep. 2020.
\bibitem{wu2018intelligent}
Q.~Wu and R.~Zhang, ``Intelligent Reflecting Surface Enhanced Wireless Network:
  Joint Active and Passive Beamforming Design,'' in \emph{Proc. IEEE GLOBECOM}, Abu Dhabi, UAE, Dec. 2018. 
\bibitem{Huang2019}
C.~Huang, A.~Zappone, G.~C. Alexandropoulos, M.~Debbah, and C.~Yuen, ``Reconfigurable Intelligent Surfaces for Energy Efficiency in Wireless Communication,'' \emph{IEEE Transactions on Wireless Communications}, vol.~18, no.~8, pp. 4157-4170, Aug. 2019. 
\bibitem{guo2020weighted}
H.~Guo, Y.-C. Liang, J.~Chen, and E.~G. Larsson, ``Weighted Sum-Rate Maximization for Reconfigurable Intelligent Surface Aided Wireless Networks,'' \emph{IEEE Transactions on Wireless Communications}, vol. 19, no. 5, pp. 3064-3076, May 2020. 
\bibitem{zhang2020reconfigurable}
H.~Zhang, B.~Di, L.~Song, and Z.~Han, ``Reconfigurable Intelligent Surfaces Assisted Communications with Limited Phase Shifts: How Many Phase Shifts are Enough?'' \emph{IEEE Transactions on Vehicular Technology}, vol.~69, no.~4, pp. 4498-4502, Apr. 2020.
 \bibitem{SLi}
S. Li, B. Duo, X. Yuan, Y. Liang, and M. Di Renzo, ``Reconfigurable Intelligent Surface Assisted UAV Communication: Joint Trajectory Design and Passive Beamforming," \emph{IEEE Wireless Communications Letters}, vol. 9, no. 5, pp. 716-720, May 2020.
\bibitem{abeywickrama}
S.~Abeywickrama, R.~Zhang, Q.~Wu, and C.~Yuen, ``Intelligent Reflecting Surface: Practical Phase Shift Model and Beamforming Optimization,''\emph{IEEE Transactions on Communications}, vol. 68, no. 9, pp. 5849-5863, Sep. 2020.
\bibitem{yu2019enabling}
X.~Yu, D.~Xu, and R.~Schober, ``Enabling Secure Wireless Communications via Intelligent Reflecting Surfaces,'' in \emph{Proc. IEEE GLOBECOM}, Waikoloa, HI, Dec. 2019.
\bibitem{9110869}
C. Huang, G. C. Alexandropoulos, C. Yuen and M. Debbah, ``Indoor Signal Focusing with Deep Learning Designed Reconfigurable Intelligent Surfaces,'' in \emph{Proc. IEEE SPAWC}, Cannes, France, Jul. 2019.
\bibitem{GCA1}
G. C. Alexandropoulos, S. Samarakoon, M. Bennis, and M. Debbah, ``Phase Configuration Learning in Wireless Networks with Multiple Reconfigurable Intelligent Surfaces,'' in \emph{Proc. IEEE GLOBECOM}, Taipei, Taiwan, Dec. 2020.
\bibitem{add2} 
C. Liu, X. Liu, D. W. K. Ng and J. Yuan, ``Deep Residual Learning for Channel Estimation in Intelligent Reflecting Surface-Assisted Multi-User Communications,'' \emph{IEEE Transactions on Wireless Communications}, 2021.
\bibitem{yang2020computation}
B.~Yang, X.~Cao, J.~Bassey, X.~Li, and L.~Qian, ``Computation Offloading in Multi-Access Edge Computing: A Multi-Task Learning Approach,'' \emph{IEEE Transactions on Mobile Computing}, vol. 20, no. 9, pp. 2745-2762, Sep. 2021.
\bibitem{XCao2}
X. Cao, B. Yang, C. Huang, C. Yuen, M. Di Renzo, D. Niyato, and Z. Han, ``Reconfigurable Intelligent Surface-Assisted Aerial-Terrestrial Communications via Multi-Task Learning," \emph{IEEE Journal on Selected Areas in Communications}, Jun. 2021.
\bibitem{JHu}
J.~Hu, H.~Zhang, B.~Di, L.~Li, K.~Bian, L.~Song, Y.~Li, Z.~Han, and H.~V. Poor, ``Reconfigurable Intelligent Surface Based RF Sensing: Design, Optimization, and Implementation," \emph{IEEE Journal on Selected Areas in Communications}, vol. 38, no. 11, pp. 2700-2716, Nov. 2020.
\bibitem{TBai}
T. Bai, C. Pan, Y. Deng, M. Elkashlan, A. Nallanathan, and L. Hanzo, ``Latency Minimization for Intelligent Reflecting Surface Aided Mobile Edge Computing," \emph{IEEE Journal on Selected Areas in Communications}, vol. 38, no. 11, pp. 2666-2682, Nov. 2020.
\bibitem{XCao3}
X. Cao, B. Yang, C. Huang, C. Yuen, Y. Zhang, D. Niyato, and Z. Han, ``Converged Reconfigurable Intelligent Surface and Mobile Edge Computing for Space Information Networks,"  \emph{IEEE Network Magazine}, Feb. 2021.
\bibitem{BYang}  
B. Yang, X. Cao, C. Huang, C. Yuen, L. Qian, and M. Di Renzo, ``Intelligent Spectrum Learning for Wireless Networks with Reconfigurable Intelligent Surfaces," \emph{EEE Transactions on Vehicular Technology}, vol. 70, no. 4, pp. 3920-3925, Apr. 2021. 
\bibitem{yang2020intelligent}
Y.~Yang, B.~Zheng, S.~Zhang, and R.~Zhang, ``Intelligent Reflecting Surface Meets OFDM: Protocol Design and Rate Maximization,'' \emph{IEEE Transactions on Communications}, vol. 68, no. 7, pp. 4522-4535, Jul. 2020.
\bibitem{MJung}  
M. Jung, W. Saad, M. Debbah, and C. S. Hong, ``On the Optimality of Reconfigurable Intelligent Surfaces (RISs): Passive Beamforming, Modulation, and Resource Allocation," \emph{IEEE Transactions on Wireless Communications}, Feb. 2021.
\bibitem{ZDing}
Z. Ding and H. V. Poor, ``A Simple Design of IRS-NOMA Transmission," \emph{IEEE Communications Letters}, vol. 24, no. 5, pp. 1119-1123, May 2020.
\bibitem{LLV}
L. Lv, Q. Wu, Z. Li, Z. Ding, N. Al-Dhahir and J. Chen, ``Covert Communication in Intelligent Reflecting Surface-Assisted NOMA Systems: Design, Analysis, and Optimization," \emph{IEEE Transactions on Wireless Communications}, Aug. 2021.
\bibitem{ZDing1}
Z. Ding, R. Schober, and H. V. Poor, ``On the Impact of Phase Shifting Designs on IRS-NOMA," \emph{IEEE Wireless Communications Letters}, vol. 9, no. 10, pp. 1596-1600, Oct. 2020.
\bibitem{XMu}
X. Mu, Y. Liu, L. Guo, J. Lin, and R. Schober, ``Joint Deployment and Multiple Access Design for Intelligent Reflecting Surface Assisted Networks," \emph{IEEE Transactions on Wireless Communications}, May 2021.
\bibitem{WNi}
W. Ni, X. Liu, Y. Liu, H. Tian, and Y. Chen, ``Resource Allocation for Multi-Cell IRS-Aided NOMA Networks,"  \emph{IEEE Transactions on Wireless Communications}, vol. 20, no. 7, pp. 4253-4268, Jul. 2021.
\bibitem{XCao1}
X. Cao, B. Yang, H. Zhang, C. Huang, C. Yuen, and Z. Han, ``Reconfigurable Intelligent Surface-Assisted MAC for Wireless Networks: Protocol Design, Analysis, and Optimization,'' \emph{IEEE Internet of Things Journal}, Mar. 2021.
\bibitem{ZDing2}
Z. Ding, R. Schober, P. Fan and H. V. Poor, ``OTFS-NOMA: An Efficient Approach for Exploiting Heterogenous User Mobility Profiles,'' \emph{IEEE Transactions on Communications}, vol. 67, no. 11, pp. 7950-7965, Nov. 2019.
\bibitem{XMu1}
X. Mu, Y. Liu, L. Guo, J. Lin and R. Schober, ``Intelligent Reflecting Surface Enhanced Indoor Robot Path Planning: A Radio Map-Based Approach," \emph{IEEE Transactions on Wireless Communications}, vol. 20, no. 7, pp. 4732-4747, Jul. 2021.
\bibitem{SLin}
S. Lin, B. Zheng, G. C. Alexandropoulos, M. Wen, M. Di Renzo, and F. Chen, ``Reconfigurable Intelligent Surfaces with Reflection Pattern Modulation: Beamforming Design, Channel Estimation, and Achievable rate Analysis,'' \emph{IEEE Transactions on Wireless Communications}, vol. 20, no. 2, pp. 741–754, Feb. 2021.
\bibitem{GCA}
G. C. Alexandropoulos and E. Vlachos, ``A Hardware Architecture for Reconfigurable Intelligent Surfaces with Minimal Active Elements for Explicit Channel Estimation,'' in \emph{Proc. IEEE ICASSP}, Barcelona, Spain, May 2020.
\bibitem{SV}
S. Venkatesh, et al., ``A High-Speed Programmable and Scalable Terahertz Holographic Metasurface Based on Tiled CMOS Chips,'' \emph{Nature Electronics,} vol. 3, pp. 785-793, Dec. 2020.
\bibitem{Bianchi}
G.~Bianchi, ``Performance Analysis of The IEEE 802.11 Distributed Coordination Function,'' \emph{IEEE Journal on Selected Areas in Communications}, vol.~18, no.~3, pp. 535-547, Mar. 2000.
\bibitem{Thilina}
K. G. M. Thilina, E. Hossain, and D. I. Kim, ``DCCC-MAC: A Dynamic Common-Control-Channel-Based MAC Protocol for Cellular Cognitive Radio Networks," \emph{IEEE Transactions on Vehicular Technology}, vol. 65, no. 5, pp. 3597-3613, May 2016.
\end{thebibliography}
\end{document}